\documentclass[
final, nomarks
]{dmtcs-episciences}

\usepackage[utf8]{inputenc}
\usepackage{subfigure}
\usepackage{amsmath}
\usepackage{amsthm}
\usepackage{amssymb}
\usepackage{enumerate}
\usepackage{cases}

\newtheorem{theorem}{Theorem}
\newtheorem{lemma}{Lemma}
\newtheorem{claimX}{Claim}
\newtheorem{propertyX}{Property}

\newcommand{\rephrase}[3]{\noindent\textbf{#1 #2}.~\emph{#3}}

\author[Giuseppe {Di Battista} and Fabrizio Frati]{Giuseppe {Di Battista} \and Fabrizio Frati}
\title[On the Resolution of Planar Drawings and Morphs]{From Tutte to Floater and Gotsman: \\On the  Resolution of Planar Straight-line Drawings and Morphs\thanks{This research was supported, in part, by MUR of Italy (PRIN Project no.\ 2022ME9Z78 –
		NextGRAAL and Project no.\ 2022TS4Y3N – EXPAND). A preliminary version of this paper appeared at the 29th International Symposium on Graph Drawing and Network Visualization ({GD} 2021)~\cite{DBLP:conf/gd/BattistaF21}.}}
\affiliation{Department of Civil, Computer Science and Aeronautical Technologies Engineering, Roma Tre University, Italy}
\keywords{Planar graphs, straight-line drawings, morphing, resolution}
\begin{document}
\publicationdata{vol. 27:2}{2025}{6}{10.46298/dmtcs.12439}{2023-10-19; 2023-10-19; 2024-08-27; 2025-02-12}{2025-02-20}
\maketitle
\begin{abstract}
The algorithm of Tutte for constructing convex planar straight-line drawings and the algorithm of Floater and Gotsman for constructing planar straight-line morphs are among the most popular graph drawing algorithms. In this paper, focusing on maximal plane graphs, we prove upper and lower bounds on the resolution of the planar straight-line drawings produced by Floater's algorithm, which is a broad generalization of Tutte's algorithm. Further, we use such results in order to prove a lower bound on the resolution of the drawings of maximal plane graphs produced by Floater and Gotsman's morphing algorithm. Finally, we show that such a morphing algorithm might produce drawings with exponentially-small resolution, even when transforming drawings with polynomial resolution.  
\end{abstract}

	\section{Introduction}

In 1963 Tutte~\cite{t-hdg-63} presented an algorithm to construct convex planar straight-line drawings of $3$-connected plane graphs. The algorithm is very simple: Given any convex polygon representing the outer cycle of the graph, place each internal vertex at the barycenter of its neighbors. This results in a system of linear equations, whose variables are the coordinates of the internal vertices, which has a unique solution; quite magically, this solution corresponds to a planar straight-line drawing of the graph in which the faces are delimited by convex polygons. We call any drawing obtained by an application of Tutte's algorithm a \emph{T-drawing}. Tutte's algorithm is one of the most famous graph drawing algorithms; notably, it has spurred the research on the practical graph drawing algorithms that are called \emph{force-directed methods}~\cite{dett-gd-99,e-hgd-84,k-fd-13}.

A far-reaching generalization of Tutte's algorithm was presented by Floater~\cite{f-psa-97,f-pt-98} (and, in a similar form, by Linial, Lov{\'{a}}sz, and Wigderson~\cite{llw-rb-88}). Namely, one can place each internal vertex at \emph{any} convex combination (with positive coefficients) of its neighbors; the resulting system of equations still has a unique solution that corresponds to a convex planar straight-line drawing of the graph. Formally, let $G$ be a $3$-connected plane graph and let $\mathcal P$ be a convex polygon representing the outer cycle of~$G$. Further, for each internal vertex $v$ of $G$ and for each neighbor $u$ of $v$, let $\lambda_{vu}>0$ be a real value such that $\sum_{u\in \mathcal N(v)} \lambda_{vu} = 1$, where $\mathcal N(v)$ denotes the set of neighbors of $v$. For each internal vertex $v$ of $G$, consider the two equations: 

\begin{center}
	\begin{tabular}{p{7cm} p{7cm}}
		\begin{equation}
			x(v)=\sum_{u\in \mathcal N(v)} (\lambda_{vu}\cdot x(u)) \label{eq:x}
		\end{equation} 
		&
		\begin{equation}
			y(v)=\sum_{u\in \mathcal N(v)} (\lambda_{vu}\cdot y(u)) \label{eq:y}
		\end{equation}
	\end{tabular}
\end{center}
where $x(v)$ and $y(v)$ denote the $x$ and $y$-coordinates of a vertex $v$, respectively. This results in a system of $2N$ equations in $2N$ variables, where $N$ is the number of internal vertices of $G$, which has a unique solution~\cite{f-psa-97,f-pt-98}. This solution corresponds to a convex planar straight-line drawing, which can hence be represented by a pair $(\Lambda,\mathcal P)$, where $\mathcal P$ is the prescribed convex polygon and $\Lambda$ is a \emph{coefficient matrix}. This matrix has a row for each internal vertex of $G$ and a column for each (internal or external) vertex of $G$; further, an element of the matrix whose row corresponds to a vertex $v$ and whose column corresponds to a vertex $u$ is the coefficient $\lambda_{vu}$ if $(v,u)$ is an edge of $G$ and $0$ otherwise. We call any drawing resulting from an application of Floater's algorithm an \emph{F-drawing}. Notice that F-drawings are extensively used for surface parameterization and reconstruction in computer graphics, in multiresolution problems, and in texture mapping; see, e.g.,~\cite{ggt-dof-06,gjm-27-21,cpv-tbm-03,warv-asf-02}. Similar types of drawings have been studied for constructing three-dimensional representations of polytopes~\cite{ds-esp-17,mrs-sge-11,DBLP:journals/jgaa/Schulz11}. Further, \emph{every} convex planar straight-line drawing of a $3$-connected plane graph (and, in particular, every planar straight-line drawing of a maximal plane graph) is an F-drawing $(\Lambda,\mathcal P)$, for a suitable choice of $\Lambda$ and $\mathcal P$~\cite{f-psa-97,f-pt-98,f-mvc-03,fg-hmti-99}. 

Floater and Gotsman~\cite{fg-hmti-99} devised a simple and yet powerful application of the above drawing technique to the construction of planar straight-line morphs. Given a graph $G$ and given two convex planar straight-line drawings $\Gamma_0$ and $\Gamma_1$ of $G$ with the same polygon $\mathcal P$ representing the outer cycle, construct two coefficient matrices $\Lambda_0$ and $\Lambda_1$ such that $(\Lambda_0,\mathcal P)=\Gamma_0$ and $(\Lambda_1,\mathcal P)=\Gamma_1$. Now, a morph $\mathcal M$ between $\Gamma_0$ and $\Gamma_1$, that is, a continuous transformation of $\Gamma_0$ into $\Gamma_1$, can be obtained as follows. For each $t\in [0,1]$, construct a coefficient matrix $\Lambda_t$ as $(1-t)\cdot \Lambda_0+t\cdot \Lambda_1$; in particular, for each internal vertex $v$ of $G$ and for each neighbor $u$ of $v$, the element $\lambda^t_{vu}$ of $\Lambda_t$ whose row corresponds to $v$ and whose column corresponds to $u$ is equal to $\lambda^t_{vu}=(1-t)\cdot \lambda^0_{vu} + t\cdot \lambda^1_{vu}$. Then the morph between $\Gamma_0$ and $\Gamma_1$ is simply defined as $\mathcal M = \{\Gamma_t=(\Lambda_t,\mathcal P) : t\in [0,1]\}$; since $\Gamma_t$ is an F-drawing, for any $t\in [0,1]$, this algorithm guarantees that every drawing in $\mathcal M$ is a convex planar straight-line drawing. We call any morph constructed by an application of Floater and Gotsman's algorithm an \emph{FG-morph}. The algorithm of Floater and Gotsman is perhaps the most popular graph morphing algorithm; extensions, refinements, and limits of the method have been discussed, e.g., in~\cite{a-ramm-02,i-vfg-14,swln-tm-03,gs-cmcpt-01}.

In this paper, we study the resolution of T-drawings, F-drawings, and FG-morphs of maximal plane graphs. The \emph{resolution} is perhaps the most studied aesthetic criterion for the readability of a graph drawing. Measuring the resolution of a drawing can be done in many ways, for example by considering the area of the drawing, if the vertices have integer coordinates, or by bounding the smallest angle in the drawing. We adopt a natural definition of resolution, namely the ratio between the smallest and the largest distance between two (distinct, non-incident, and non-adjacent) geometric objects representing vertices or edges. Eades and Garvan~\cite{eg-dspg-95} proved that T-drawings of $n$-vertex maximal plane graphs might have $1/2^{\Omega(n)}$ resolution; this was independently observed by Chambers et al.~\cite{cegl-dgp-12}. Furthermore, an $1/2^{O(n)}$ lower bound for the resolution of T-drawings {\em in which one is allowed to choose the polygon representing the outer cycle} has been proved by~\cite{mrs-sge-11}; see also~\cite{rg-rsp-96} and Section~\ref{se:algebra}, where we discuss the techniques used in~\cite{mrs-sge-11,rg-rsp-96}. It is unclear, a priori, whether the worst-case resolution of T-drawings, F-drawings, and FG-morphs can be expressed as \emph{any} function of natural parameters\footnote{The number $n$ of vertices of the graph is not the only parameter that needs to be taken into account in order to bound the resolution of T-drawings, F-drawings, and FG-morphs. Indeed, for T-drawings and F-drawings, one might have to consider the resolution of the prescribed polygon $\mathcal P$ (which is a triangle, since we deal with maximal plane graphs) and the values of the coefficient matrix $\Lambda$. Further, for FG-morphs, one might have to consider the resolution of the input drawings $\Gamma_0$ and $\Gamma_1$.} representing the input size and resolution. 	




We prove the following results\footnote{At a first glance, our use of the $O(\cdot)$ and $\Omega(\cdot)$ notation in the paper seems to be inverted. For example, Theorem~\ref{th:lower-bound-gt-drawings} shows a bound of $r\cdot \lambda^{O(n)}$ on the resolution of F-drawings. This is a lower bound, and not an upper bound, given that $\lambda<1$. Indeed, the $O(\cdot)$ notation indicates that the exponent has a value which is \emph{at most something}, hence the entire power has a value which is \emph{at least something} (smaller than one).}. First, we show a lower bound on the resolution of F-drawings (and thus on the resolution of T-drawings).

\begin{theorem} \label{th:lower-bound-gt-drawings}
Let $\Gamma=(\Lambda,\Delta)$ be an F-drawing of an $n$-vertex maximal plane graph $G$, where $n\geq 4$. The resolution of $\Gamma$ is larger than or equal to $\frac{r}{2} \cdot \left(\frac{\lambda}{3}\right)^n \in r\cdot \lambda^{O(n)}$, where $\lambda$ is the smallest positive coefficient in the coefficient matrix $\Lambda$ and $r$ is the resolution of the prescribed triangle $\Delta$.
\end{theorem}

Second, we prove that the bound for the exponent in Theorem~\ref{th:lower-bound-gt-drawings} is asymptotically tight for the family of graphs introduced by Eades and Garvan~\cite{eg-dspg-95}.

\begin{theorem} \label{th:upper-bound-gt-drawings}
There is a class of maximal plane graphs $\{G_n:n=5,6,\dots\}$, where $G_n$ has $n$ vertices, with the following property. For any $0<\lambda\leq \frac{1}{4}$ and $0<r\leq \frac{\sqrt 3}{2}$, there exist a triangle $\Delta$ with resolution $r$ and a coefficient matrix $\Lambda$ for $G_n$ whose smallest positive coefficient is $\lambda$ such that the F-drawing $(\Lambda,\Delta)$ of $G_n$ has resolution in $r\cdot \lambda^{\Omega(n)}$.
\end{theorem}

We remark that algorithms are known for constructing planar straight-line drawings of maximal plane graphs~\cite{dpp-htgg-90,s-epgg-90} and even convex planar straight-line drawings of $3$-connected plane graphs~\cite{br-scdpg-06,cg-cdg-96} with polynomial resolution. Indeed, drawings on the grid with polynomial area have polynomially bounded resolution. 

Third, we use Theorem~\ref{th:lower-bound-gt-drawings} in order to prove a lower bound on the resolution of FG-morphs.

\begin{theorem} \label{th:lower-bound-morph} 
Let $\Gamma_0$ and $\Gamma_1$ be any two planar straight-line drawings of the same $n$-vertex maximal plane graph $G$ such that the outer faces of $\Gamma_0$ and $\Gamma_1$ are delimited by the same triangle $\Delta$.
There exists an FG-morph $\mathcal M=\{\Gamma_t: t\in[0,1]\}$ between $\Gamma_0$ and $\Gamma_1$ such that, for each $t\in[0,1]$, the resolution of $\Gamma_t$ is larger than or equal to $\left(r/n\right)^{O(n)}$, where $r$ is the minimum between the resolution of $\Gamma_0$ and $\Gamma_1$.
\end{theorem}

We show that Theorem~\ref{th:lower-bound-morph} can be used in order to ``approximate'' FG-morphs with piecewise linear morphs of finite complexity. A \emph{piecewise linear morph} consists of a sequence of linear morphs, in which vertices move at uniform speed along straight-line segments. It is intuitive that one can mimic an FG-morph by a sequence of ``suitably short'' linear morphs, each connecting two drawings of the FG-morph, so that planarity is maintained at all times. Our result, whose precise statement is deferred to Section~\ref{se:approximating-FG}, shows that a finite number of linear morphs suffices to design such a discretization of the FG-morph.

Finally, we prove that FG-morphs might produce drawings with exponentially small resolution, even if they transform drawings with polynomial resolution.

\begin{theorem} \label{th:upper-bound-morphing}
For every $n\geq 6$ multiple of $3$, there exist an $n$-vertex maximal plane graph $G$ and two planar straight-line drawings $\Gamma_0$ and $\Gamma_1$ of $G$ such that:  
\begin{enumerate}[(R1)]
	\item the outer faces of $\Gamma_0$ and $\Gamma_1$ are delimited by the same triangle $\Delta$; 
	\item the resolution of both $\Gamma_0$ and $\Gamma_1$ is larger than or equal to $c/n^2$, for some constant $c$; and
	\item any FG-morph between $\Gamma_0$ and $\Gamma_1$ contains a drawing whose resolution is in $1/2^{\Omega(n)}$.
\end{enumerate}
\end{theorem}

We remark that the construction of planar straight-line morphs with high resolution has been attracting an increasing attention~\cite{bbd-hmt-19,bhl-msd-19,ddf-upm-20}. A main open question in the area is whether polynomial resolution can be guaranteed in a planar straight-line morph between any two given drawings (with polynomial resolution) of a plane graph. Theorem~\ref{th:upper-bound-morphing} shows that Floater and Gotsman's algorithm cannot be used to settle this question in the positive.






\section{Preliminaries}

We assume familiarity with graph drawing~\cite{dett-gd-99} and introduce preliminary properties and lemmata.

We denote by $V(G)$ and $E(G)$ the vertex and edge sets of a graph $G$, respectively. For a vertex $v\in V(G)$, we denote by $\mathcal N_G(v)$ (or by $\mathcal N(v)$ when the graph is clear from the context) the set of neighbors of $v$ in $G$; the \emph{degree} of $v$ is $|\mathcal N(v)|$. Throughout the paper, by the term \emph{cycle} we always refer to a simple cycle. We denote by $n$ the number of vertices of the considered graphs; in the following, we always assume that $n\geq 4$.

A graph is \emph{biconnected} if the removal of any vertex leaves the graph connected.

A \emph{drawing} of a graph represents each vertex as a distinct point in the plane and each edge as a Jordan arc connecting the points representing the end-vertices of the edge, so that no point representing a vertex lies in the interior of a Jordan arc representing an edge. A drawing of a graph is \emph{planar} if no two Jordan arcs representing edges intersect, except at common end-points. A planar drawing of a graph partitions the plane into connected regions, called \emph{faces}. The only unbounded face is the \emph{outer face}, while the bounded faces are \emph{internal}. In a planar drawing of a biconnected planar graph every face is delimited by a cycle, which is a \emph{facial cycle}; the facial cycle bounding the outer face is called \emph{outer cycle}. Two planar drawings $\Gamma_1$ and $\Gamma_2$ of a biconnected planar graph $G$ are \emph{combinatorially equivalent} if: (1) for each vertex $v\in V(G)$, the clockwise order of the edges incident to $v$ is the same in $\Gamma_1$ as in $\Gamma_2$; (2) the clockwise order of the vertices along the outer cycles of $\Gamma_1$ and $\Gamma_2$ coincide. A \emph{plane embedding} is an equivalence class of planar drawings. A \emph{plane graph} is a planar graph with an associated plane embedding. A \emph{maximal plane graph} is a plane graph to which no edge can be added without losing planarity or simplicity. The assumption $n\geq 4$ implies that every vertex of an $n$-vertex maximal plane graph has degree at least $3$. 

Let $G$ be a plane graph. When referring to a planar drawing $\Gamma$ of $G$, we always mean that $\Gamma$ is in the equivalence class associated to $G$. All the planar drawings of $G$ have the same set of facial cycles, hence we often talk about the ``faces of $G$'', meaning the faces of any planar drawing in the equivalence class associated to $G$. A subgraph $G'$ of $G$ is associated with a plane embedding ``inherited'' from the one of $G$, as follows. Consider a planar drawing $\Gamma$ of $G$ and the planar drawing $\Gamma'$ of $G'$ obtained from $\Gamma$ by removing the vertices in $V(G)\setminus V(G')$ and the edges in $E(G)\setminus E(G')$; then the plane embedding of $G'$ is the equivalence class of planar drawings of $G'$ the drawing $\Gamma'$ belongs to. We say that $G$ is \emph{internally-triangulated} if every internal face of $G$ is delimited by a $3$-cycle. A vertex of $G$ is \emph{external} if it is incident to the outer face of $G$, it is \emph{internal} otherwise. The sets of internal and external vertices of $G$ are denoted by $\mathcal I_G$ and $\mathcal O_G$, respectively. Consider a cycle $\mathcal C$ of $G$. An \emph{external chord} of $\mathcal C$ is an edge of $G$ that connects two vertices of $\mathcal C$, that does not belong to $\mathcal C$, and that lies outside $\mathcal C$ in $G$. The \emph{subgraph of $G$ inside $\mathcal C$} is composed of the vertices and edges that lie inside or on the boundary of $\mathcal C$. The following is easy to observe. 

\begin{propertyX}\label{pr:subgraph-inside-cycle}
Let $G$ be a maximal plane graph and let $\mathcal C$ be a cycle of $G$. The subgraph of $G$ inside $\mathcal C$ is biconnected and internally-triangulated.
\end{propertyX} 



\subsection{Distances, resolution, and F-drawings} \label{sse:fg}

A drawing of a graph is \emph{straight-line} if every edge is represented by a straight-line segment. In a straight-line drawing of a graph, by \emph{geometric object} we mean a point representing a vertex or a straight-line segment representing an edge. We often call ``vertex''  or ``edge'' both the combinatorial object and the corresponding geometric object. Two geometric objects in a planar straight-line drawing of a graph are \emph{separated} if they share no point. By the planarity of the drawing, two geometric objects are hence separated if and only if they are distinct vertices, or non-adjacent edges, or a vertex and a non-incident edge. The \emph{distance} $d_{\Gamma}(o_1,o_2)$ between two separated geometric objects $o_1$ and $o_2$ in a planar straight-line drawing $\Gamma$ of a graph $G$ is the minimum Euclidean distance between any point of $o_1$ and any point of $o_2$. For two vertices $u$ and $v$ of $G$, we denote by $d^\updownarrow_{\Gamma}(u,v)$ the vertical distance between $u$ and $v$ in $\Gamma$, that is, the absolute value of the difference between their $y$-coordinates. The \emph{resolution} of $\Gamma$ is the ratio between the distance of the closest separated geometric objects and the distance of the farthest separated geometric objects in $\Gamma$. Throughout the paper, we denote by $\delta$ the distance of the closest separated geometric objects in any considered drawing. 

Let $\mathcal R$ be a finite connected subset of $\mathbb R^n$ and let $z_1,\dots,z_n$ be the coordinates of $\mathbb R^n$. For any $l=1,\dots,n$, the \emph{$z_l$-extent} of $\mathcal R$ is the maximum $z_l$-coordinate of any point of $\mathcal R$ minus the minimum $z_l$-coordinate of any point of $\mathcal R$. 

Throughout the paper, we denote by $(\Lambda,\Delta)$ an F-drawing of a maximal plane graph $G$, where $\Lambda$ is a coefficient matrix for $G$ and $\Delta$ is a triangle representing the outer cycle of $G$. Also, we denote by $\lambda$ the smallest positive coefficient in $\Lambda$. Note that a T-drawing is an F-drawing in which, for every internal vertex $v$ and every neighbor $u$ of $v$, we have $\lambda_{vu}=1/|\mathcal N(v)|$. 

Next, we present a tool that we often use in order to translate and rotate F-drawings. The lemma shows that, given an F-drawing $(\Lambda,\Delta)$, if we construct a new F-drawing using the coefficient matrix $\Lambda$ and an affinely transformed copy of $\Delta$, we obtain a drawing that is an affinely transformed copy of $(\Lambda,\Delta)$, under the same affine transformation.

\begin{lemma} \label{le:rotating}
Let $(\Lambda,\Delta)$ be an F-drawing of a maximal plane graph $G$ and let $\Delta'$ be the triangle obtained by applying an affine transformation $\sigma$ to $\Delta$. Then the F-drawing $(\Lambda,\Delta')$ coincides with the drawing obtained by applying $\sigma$ to the drawing $(\Lambda,\Delta)$.
\end{lemma}

\begin{proof}
The affine transformation $\sigma$ maps any point $(x,y)$ in the plane to the point $(\overline{x},\overline{y})= (a\cdot x +b\cdot y+c, d\cdot x+e\cdot y+f)$, for some scalars $a$, $b$, $c$, $d$, $e$, and $f$. We need to show that, for every internal vertex $v$ of $G$, we have 
$$\overline {x(v)}=\sum_{u\in \mathcal N(v)} (\lambda_{vu}\cdot \overline{x(u)})$$ \noindent and $$\overline{y(v)}=\sum_{u\in \mathcal N(v)} (\lambda_{vu}\cdot \overline{y(u)}).$$

We only prove the former equality, as the proof for the latter is analogous. By definition of $\sigma$, we have $\overline {x(w)}=a\cdot x(w) +b\cdot y(w)+c$, for each vertex $w$, and thus
\begin{eqnarray*}
	\sum_{u\in \mathcal N(v)} (\lambda_{vu}\cdot \overline{x(u)}) &=&\sum_{u\in \mathcal N(v)} (\lambda_{vu}\cdot (a\cdot x(u) +b\cdot y(u)+c))\\
	&=&a \cdot \sum_{u\in \mathcal N(v)} (\lambda_{vu}\cdot x(u)) +b \cdot \sum_{u\in \mathcal N(v)} (\lambda_{vu}\cdot y(u))+c \cdot \sum_{u\in \mathcal N(v)} \lambda_{vu}\\
	&=& a \cdot x(v) + b \cdot y(v) +c=\overline {x(v)},
\end{eqnarray*}

where the second to last equality follows from Equations~\ref{eq:x} and~\ref{eq:y} applied to the drawing $(\Lambda,\Delta)$ and from $\sum_{u\in \mathcal N(v)} \lambda_{vu}=1$. 
\end{proof}

We sometimes use the following elementary property.

\begin{propertyX} \label{pr:minimum-lambda}
Let $(\Lambda,\Delta)$ be an F-drawing of a maximal plane graph $G$ with at least four (five) vertices. The smallest positive coefficient $\lambda$ in $\Lambda$ is at most equal to $1/3$ (resp.\ $1/4$).
\end{propertyX}

\begin{proof}
Since $G$ has at least four (five) vertices, it contains an internal vertex $v$ with degree at least $3$ (resp.\ at least $4$). Since the coefficients expressing the position of $v$ with respect to the ones of its neighbors are all positive and sum up to $1$, at least one of such coefficients is smaller than or equal to $1/3$ (resp.~$1/4$). 
\end{proof}

\subsection{The resolution of triangles}

We now show some properties on the resolution of triangles. We start with the following observation.


\begin{propertyX} \label{pro:resolution-triangle}
The resolution of a triangle is smaller than or equal to $\frac{\sqrt 3}{2}$ and there exist triangles with this resolution. 
\end{propertyX}

\begin{proof}
First, note that an equilateral triangle has resolution $\frac{\sqrt 3}{2}$.

Consider any triangle $\Delta$ and let $\alpha\geq 60^{\circ}$ be its largest angle. Let $v$ be the vertex of $\Delta$ incident to such an angle, let $\ell$ be the length of the side $s$ of $\Delta$ opposite to $v$, and let $h$ be the height of $\Delta$ with respect to~$s$. Note that $h$ and $\ell$ are the smallest and the largest distance between two separated geometric objects of $\Delta$, respectively, hence the resolution of $\Delta$ is $h/\ell$.

The altitude of $\Delta$ through $v$ partitions $\alpha$ into two smaller angles $\alpha_1$ and $\alpha_2$. Simple trigonometric considerations show that $\ell=h\cdot(\tan(\alpha_1) +\tan(\alpha_2))$. By substituting $\alpha_2=\alpha-\alpha_1$ and by deriving with respect to $\alpha_1$, we get that the minimum of the function $f(\alpha_1):=\tan(\alpha_1) +\tan(\alpha-\alpha_1)$ is achieved when $\alpha_1=\alpha/2$. Hence, $\frac{h}{\ell}= \frac{1}{\tan(\alpha_1) +\tan(\alpha_2)}\leq \frac{1}{2\tan(\alpha/2)}\leq \frac{1}{2\tan(30^{\circ})}=\frac{\sqrt 3}{2}$. \end{proof}

We now prove the following.

\begin{lemma} \label{le:triangle-properties}
Let $\Delta$ be a triangle with resolution equal to $r$. Let $s$ be any side of $\Delta$, let $\ell$ be its length, and let $h$ be the height of $\Delta$ with respect to $s$. Then $h/\ell \geq r$.
\end{lemma}

\begin{proof}
Let $A$ be the area of $\Delta$. Let $s$ be any side of $\Delta$, let $\ell$ be its length, and let $h$ be the height of $\Delta$ with respect to $s$. Also, let $s_1$ be the longest side of $\Delta$, let $\ell_1$ be its length, and let $h_1$ be the height of $\Delta$ with respect to $s_1$. Note that $\ell_1\geq \ell$ and $h_1\leq h$. Then we have $h/\ell=2A/\ell^2\geq 2A/\ell_1^2=(\ell_1\cdot h_1)/\ell_1^2=h_1/\ell_1=r$.
\end{proof}

\begin{lemma} \label{le:resolution-xy-triangle}
Let $\Delta$ be a triangle with resolution equal to $r$ and $y$-extent equal to $Y$. Then the $x$-extent $X$ of $\Delta$ is at most $Y/r$.
\end{lemma}

\begin{proof}
Let $\mathcal R$ be the smallest axis-parallel rectangle containing $\Delta$; note that $\mathcal R$ has width $X$ and height $Y$. We distinguish two cases, based on whether two or three vertices of $\Delta$ belong to the boundary of $\mathcal R$.

\begin{figure}[htb]\tabcolsep=4pt
	\centering
	\begin{tabular}{c c}
		\includegraphics[scale=1]{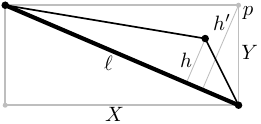}\hspace{5mm} & 				\includegraphics[scale=1]{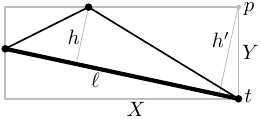}\\
		(a) \hspace{5mm} & (b) 
	\end{tabular}
	
	\caption{Illustration for the proof of Lemma~\ref{le:resolution-xy-triangle}. The side $s$ of $\Delta$ is represented by a fat line segment. In (a) two vertices of $\Delta$ belong to the boundary of $\mathcal R$, and in (b) three vertices of $\Delta$ belong to the boundary of $\mathcal R$.}
	\label{fig:triangle}
\end{figure}

{\em Case 1:} Two vertices of $\Delta$ belong to the boundary of $\mathcal R$; see Figure~\ref{fig:triangle}(a). In this case, one of the sides of $\Delta$, call it $s$, coincides with one of the diagonals of $\mathcal R$; let $\ell$ be the length of $s$. Also, let $h$ be the height of $\Delta$ with respect to $s$. Let $p$ be a vertex of $\mathcal R$ not incident to $s$, and let $h'$ be the distance between $p$ and $s$. 

Note that $h<h'$. Further, by Lemma~\ref{le:triangle-properties}, we have $h/\ell\geq r$, hence
\begin{equation}\label{eqn:triangle-1}
	h'\geq r\cdot \ell.
\end{equation}

By considering the area of $\mathcal R$ and the area of the triangle defined by $s$ and by $p$, we get

\begin{equation}\label{eqn:triangle-2}
	X\cdot Y = h'\cdot \ell.
\end{equation}

Moreover, since $\ell>X$, we get

\begin{equation}\label{eqn:triangle-3}
	\ell^2>X^2.
\end{equation}

Substituting Equation~\ref{eqn:triangle-1} into Equation~\ref{eqn:triangle-2}, we get 

\begin{equation}\label{eqn:triangle-4}
	X\cdot Y\geq r\cdot \ell^2.
\end{equation}

Substituting Equation~\ref{eqn:triangle-3} into Equation~\ref{eqn:triangle-4}, we get $X\cdot Y\geq r\cdot X^2$, from which it follows that $X\leq Y/r$, as required.

{\em Case 2:} Three vertices of $\Delta$ belong to the boundary of $\mathcal R$; see Figure~\ref{fig:triangle}(b). In this case, one of the vertices of $\Delta$, call it $t$, coincides with one of the vertices of $\mathcal R$; let $s$ be the side of $\Delta$ that is incident to $t$ and whose other end-vertex is on a vertical side of $\mathcal R$. Then the length $\ell$ of $s$ is larger than $X$, from which Equation~\ref{eqn:triangle-3} follows again. Let $h$ be the height of $\Delta$ with respect to $s$; notice that the altitude of $\Delta$ with respect to $s$ does not necessarily lie entirely inside $\mathcal R$. Furthermore, let $p$ be the vertex of $\mathcal R$ that shares a vertical side of $\mathcal R$ with $t$. Finally, let $h'$ be the distance between $p$ and $s$. 

By considering the area of $\mathcal R$ and the area of the triangle defined by $s$ and by $p$, Equation~\ref{eqn:triangle-2} follows again.

Equation~\ref{eqn:triangle-1} is also true in Case~2. Indeed, that $h\leq h'$ follows from the fact that $p$ is the farthest point of $\mathcal R$ from $s$. Moreover, that $h/\ell\geq r$ follows from Lemma~\ref{le:triangle-properties}.

Now, from Equations~\ref{eqn:triangle-1},~\ref{eqn:triangle-2}, and~\ref{eqn:triangle-3}, the bound $X\leq Y/r$ can be derived as in Case~1.
\end{proof}

\section{Lower Bound on the Resolution of F-Drawings}\label{se:lower-bound-gt-drawings}

In this section, we prove Theorem~\ref{th:lower-bound-gt-drawings}, which we restate here for the reader's convenience.

\medskip	
\rephrase{Theorem}{\ref{th:lower-bound-gt-drawings}}{%
Let $\Gamma=(\Lambda,\Delta)$ be an F-drawing of an $n$-vertex maximal plane graph $G$, where $n\geq 4$. The resolution of $\Gamma$ is larger than or equal to $\frac{r}{2} \cdot \left(\frac{\lambda}{3}\right)^n \in r\cdot \lambda^{O(n)}$, where $\lambda$ is the smallest positive coefficient in the coefficient matrix $\Lambda$ and $r$ is the resolution of the prescribed triangle $\Delta$. 
}
\medskip	

Let $\delta$ be the minimum distance between any two separated geometric objects in $\Gamma$. We start by proving that the smallest distance $\delta$ in $\Gamma$ is achieved ``inside'' an internal face of $G$. 

\begin{lemma} \label{le:first-vertex-is-internal}
Let $\Gamma=(\Lambda,\Delta)$ be an F-drawing of an $n$-vertex maximal plane graph $G$, where $n\geq 4$, and let $\delta$ be the minimum distance between any two separated geometric objects in $\Gamma$. There exist an internal vertex $v$ and an edge $e=(u_e,v_e)$ of $G$ such that:
\begin{itemize}
	\item $d_{\Gamma}(v,e)=\delta$; 
	\item $v$, $u_e$, and $v_e$ are the vertices of a triangle $T$ delimiting an internal face of $G$ in $\Gamma$; and
	\item the altitude of $T$ through $v$ lies inside $T$.
\end{itemize}
\end{lemma}

\begin{proof}
First, observe that there exist a vertex $v$ and an edge $e=(u_e,v_e)$ of $G$ such that $d_{\Gamma}(v,e)=\delta$. Indeed, the minimum distance between two separated geometric objects of $\Gamma$ is $\delta$, by assumption; each of such objects can be either a vertex or an edge. If the two objects are two vertices $u_1$ and $u_2$, then the distance between $u_1$ and any edge incident to $u_2$ is at most $\delta$; further, if the two objects are two edges $e_1$ and $e_2$, then there exists an end-vertex of one of them that is at distance $\delta$ from the other edge.

Second, we prove that $v$, $u_e$, and $v_e$ are the vertices of a triangle $T$ delimiting an internal face of $G$ in~$\Gamma$. Consider a segment $s$ of length $\delta$ between $v$ and a point $p$ of $e$; this exists because $d_{\Gamma}(v,e)=\delta$. 

\begin{itemize}
	\item Suppose first, for a contradiction, that $s$ is the straight-line segment representing an edge $e'$ of $G$. Let $f$ be any internal face of $G$ incident to $e'$ and let $T_f$ be the triangle delimiting $f$ in $\Gamma$. Then (at least) one of the heights of $T_f$ is smaller than $s$, and hence smaller than $\delta$, a contradiction.     
	\item Suppose next that $s$ does not represent an edge of $G$. If $s$ intersected a vertex or an edge in its interior, the distance between such a geometric object and $v$ would be smaller than $\delta$, which would contradict the assumptions. It follows that $s$ lies inside an internal face of $G$ that is incident to $v$, $u_e$, and $v_e$. Since $G$ is a maximal plane graph, such a face is delimited by a triangle $T$, whose vertices are $v$, $u_e$, and $v_e$.
\end{itemize} 

Third, suppose, for a contradiction, that the altitude of $T$ through $v$ does not lie inside $T$, hence the angle of $T$ at $u_e$ or $v_e$ is larger than $90^\circ$ (see Figure~\ref{fig:internal-vertex}(a)). Assume the former, as the other case is analogous. It follows that the altitude of $T$ through $u_e$ is the shortest altitude of $T$, hence the distance in $\Gamma$ between $u_e$ and the edge $(v,v_e)$ is smaller than $\delta$, a contradiction.

\begin{figure}[htb]\tabcolsep=4pt
	\centering
	\begin{tabular}{c c}
		\includegraphics[scale=0.85]{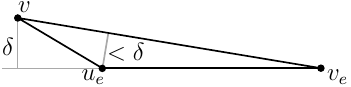}\hspace{3mm} &
		\includegraphics[scale=0.85]{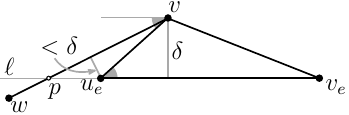}\\
		(a) \hspace{3mm} & (b) \hspace{3mm} 
	\end{tabular}
	
	\caption{Illustration for the proof of Lemma~\ref{le:first-vertex-is-internal}. (a) If the altitude of $T$ through $v$ does not lie inside $T$, then there is a height of $T$ smaller than $\delta$. (b) If $v$ is an external vertex of $G$, then the distance between one of the end-vertices of $e$ and one of the edges incident to the outer face of $G$ is smaller than $\delta$. The gray angles are equal.}
	\label{fig:internal-vertex}
\end{figure}

Finally, we prove that $v$ is an internal vertex of $G$; refer to Figure~\ref{fig:internal-vertex}(b). Suppose, for a contradiction, that $v$ is an external vertex of $G$ and let $w$ and $z$ be the two vertices of the outer cycle of $G$ different from $v$. Since $n\geq 4$ and since the $3$-cycle $(v,u_e,v_e)$ bounds an internal face of $G$, we have that the edge $e$ is not incident to the outer face of $G$. It follows that at least one of the edges $(v,w)$ and $(v,z)$ cuts the line $\ell$ through $e$, as otherwise $e$ would not lie inside the triangle delimiting the outer face of $\Gamma$. Say that $(v,w)$ cuts $\ell$ in a point $p$ and assume, without loss of generality, that $p$, $u_e$, and $v_e$ occur in this order along $\ell$. Then the distance from $u_e$ to the edge $(v,w)$ is $|\overline{vu_e}|\cdot \sin(\widehat{u_evw})$, which is smaller than  $|\overline{vu_e}|\cdot \sin(\widehat{vu_ev_e})=\delta$, given that $\widehat{u_evw}<\widehat{vu_ev_e}$, a contradiction which concludes the proof of the lemma. 
\end{proof}

By Lemma~\ref{le:rotating}, we can assume that $y(v)=0$, that $e$ is horizontal in $\Gamma$, and that $v$ lies above (the line through) $e$. We now show that the neighbors of $v$ are not ``too high'' or ``too low'' in $\Gamma$.

\begin{lemma} \label{le:first-vertex-height}
For every neighbor $u$ of $v$, we have that $d^\updownarrow_{\Gamma}(u,v)\leq \frac{\delta}{\lambda}$.
\end{lemma}

\begin{proof}
By Lemma~\ref{le:first-vertex-is-internal}, the altitude through $v$ of the triangle with vertices $v$, $u_e$, and $v_e$ intersects the edge $(u_e,v_e)$. It follows that: (i) $y(u_e)=y(v_e)=-\delta$, which implies that the statement of the lemma is satisfied for $u\in \{u_e,v_e\}$, given that $\lambda<1$; and (ii) $x(u_e)<x(v)<x(v_e)$, up to a relabeling of $u_e$ with~$v_e$.  


Next, we prove the following claim: Every neighbor $w$ of $v$ that is different from $u_e$ and $v_e$ lies on or above the horizontal line through $v$; that is, $y(w)\geq 0$, for every vertex $w\in \mathcal N(v) \setminus \{u_e,v_e\}$. By the planarity of $\Gamma$, it suffices to prove the claim for the neighbors $u$ and $z$ of $v$ such that $(v,u)$ and $(v,z)$ are the edges that follow $(v,u_e)$ in clockwise and counter-clockwise direction around $v$, respectively.

We prove that $y(u)\geq 0$, the proof that $y(z)\geq 0$ is analogous. Suppose, for a contradiction, $y(u)<0$. Let $\Delta_e$ be the triangle with vertices $u$, $v$, and $u_e$. We distinguish two cases. 

\begin{figure}[htb]\tabcolsep=4pt
	\centering
	\begin{tabular}{c c}
		\includegraphics[scale=0.85]{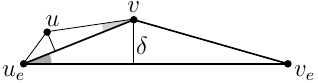}\hspace{3mm} &
		\includegraphics[scale=0.85]{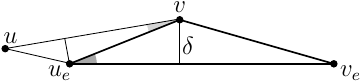}\\
		(a) \hspace{3mm} & (b) \hspace{3mm} 
	\end{tabular}
	\caption{Illustration for the proof that $y(u)\geq 0$, where $(v,u)$ is the edge that follows $(v,u_e)$ in clockwise order around $v$. (a) The case in which $\widehat{vuu_e}>\widehat{vu_eu}$. (b) The case in which $\widehat{vu_eu}>\widehat{vuu_e}$.}
	\label{fig:minimum-distance}
\end{figure}

\begin{itemize}
	\item Suppose first that the angle of $\Delta_e$ at $u$ is larger than or equal to the angle of $\Delta_e$ at $u_e$, that is, $\widehat{vuu_e}\geq \widehat{vu_eu}$ (see Figure~\ref{fig:minimum-distance}(a)). This implies that $d_\Gamma(v,u)\leq d_\Gamma(v,u_e)$ and that the altitude of $\Delta_e$ through $u$ intersects the edge $(v,u_e)$. It follows that the distance between $u$ and the edge $(v,u_e)$ is equal to $d_\Gamma(v,u) \cdot \sin (\widehat{uvu_e})<d_\Gamma(v,u_e) \cdot \sin(\widehat{vu_ev_e})=\delta$, a contradiction to the fact that $\delta$ is the minimum distance between two separated geometric objects in $\Gamma$.
	\item Suppose next that the angle of $\Delta_e$ at $u_e$ is larger than the angle of $\Delta_e$ at $u$, that is, $\widehat{vu_eu}>\widehat{vuu_e}$ (see Figure~\ref{fig:minimum-distance}(b)). This implies that the altitude of $\Delta_e$ through $u_e$ intersects the edge $(v,u)$. It follows that the distance between $u_e$ and the edge $(v,u)$ is equal to $d_\Gamma(v,u_e) \cdot \sin(\widehat{uvu_e})<d_\Gamma(v,u_e) \cdot \sin(\widehat{vu_ev_e})=\delta$, a contradiction to the fact that $\delta$ is the minimum distance between two separated geometric objects in~$\Gamma$.
\end{itemize}

This completes the proof of the claim. Now, by Equation~\ref{eq:y}, we have that $\sum_{u\in \mathcal N(v)} (\lambda_{vu}\cdot y(u))=y(v)=0$, hence $\sum_{u\in \mathcal N(v)\setminus \{u_e,v_e\}} (\lambda_{vu}\cdot y(u))=(\lambda_{vu_e}+\lambda_{vv_e})\cdot \delta$. By the above claim, for  every vertex $u\in \mathcal N(v) \setminus \{u_e,v_e\}$, we have that $\lambda_{vu}\cdot y(u)\leq (\lambda_{vu_e}+\lambda_{vv_e})\cdot \delta$, and hence $y(u)< \frac{\delta}{\lambda_{vu}}\leq \frac{\delta}{\lambda}$.
\end{proof}

We outline the proof of Theorem~\ref{th:lower-bound-gt-drawings}. By Lemma~\ref{le:first-vertex-height}, the vertex $v$ and its neighbors are contained in $\Gamma$ inside a ``narrow'' horizontal strip (see Figure~\ref{fig:v-and-neighbors}(a)). Using that as a starting point, the strategy is now to define a sequence of subgraphs of $G$, each one larger than the previous one, so that each subgraph is contained inside a narrow horizontal strip. The larger the considered graph, the larger the height of the horizontal strip, however this height only depends on $\lambda$, on $\delta$, and on the number of vertices of the considered graph. Eventually, this argument leads to a bound on the $y$-extent of $\Gamma$, and from that bound the resolution $r$ of the outer triangle $\Delta$ provides a bound on the largest distance between two separated geometric objects of $\Gamma$. The comparison of such a distance with the minimum distance $\delta$ between two separated geometric objects of $\Gamma$ allows us to derive the bound of Theorem~\ref{th:lower-bound-gt-drawings}.

We now formalize the above proof strategy. For $i \in \mathbb{N}^+$, we denote by $\mathcal H_i$ the horizontal strip  of height $h(i):=\delta\cdot\left(\frac{3}{\lambda}\right)^i$ bisected by the horizontal line through $v$.

We prove the existence of a sequence $G_1,\dots,G_k=G$ of graphs such that, for $i=1,\dots,k$, the graph~$G_i$ is a biconnected internally-triangulated plane graph that is a subgraph of $G$ satisfying Properties~(P1)--(P4) below. Let $\Gamma_i$ be the restriction of $\Gamma$ to $G_i$ and let $\mathcal C_i$ be the outer cycle of $G_i$.
\begin{enumerate}[(P1)]
\item $G_i$ has at least $i+3$ vertices;
\item $\mathcal C_i$ does not have any external chord;
\item $G_i$ is the subgraph of $G$ inside $\mathcal C_i$; and
\item $\Gamma_i$ is contained in the interior of the horizontal strip  $\mathcal H_i$.
\end{enumerate}

\begin{figure}[htb]\tabcolsep=4pt
\centering
\begin{tabular}{c c c}
	\includegraphics[scale=0.85]{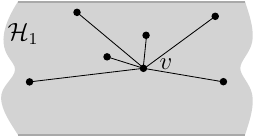}\hspace{3mm} &
	\includegraphics[scale=0.85]{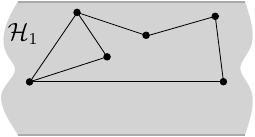}\hspace{3mm} &
	\includegraphics[scale=0.85]{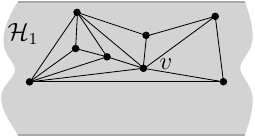}\\
	(a) \hspace{3mm} & (b) \hspace{3mm} & (c)
\end{tabular}
\caption{(a) The subgraph of $G$ composed of the edges incident to $v$. (b) The subgraph $G^v$ of $G$ induced by the neighbors of $v$. (c) The subgraph $G_1$ of $G$ inside $\mathcal C_1$.}
\label{fig:v-and-neighbors}
\end{figure}

We define $G_1$ as follows. Let $G^v$ be the subgraph of $G$ induced by the neighbors of $v$ (see Figure~\ref{fig:v-and-neighbors}(b)). Since $G$ is a maximal plane graph, we have that $G^v$ is biconnected; let $\mathcal C_1$ be the outer cycle of $G^v$. Then~$G_1$ is defined as the subgraph of $G$ inside $\mathcal C_1$ (see Figure~\ref{fig:v-and-neighbors}(c)). Property~(P3) is satisfied by construction. Further, Property~(P1) is satisfied, given that $G_1$ contains at least $4$ vertices, namely $v$ and its at least three neighbors. Property~(P2) is satisfied because $\mathcal C_1$ is the outer cycle of $G^v$ and $G^v$ is an induced subgraph of $G$. Property~(P4) is satisfied by Lemma~\ref{le:first-vertex-height}, as the distance from $v$ to the top or bottom side of $\mathcal H_1$ is $\frac{3}{2}\cdot\frac{\delta}{\lambda}$; observe that all the vertices of $G_1$ are contained in the convex hull of the neighbors of $v$. Finally, $G_1$ is biconnected and internally-triangulated, by Property~\ref{pr:subgraph-inside-cycle}. 

Assume that $G_i\neq G$; we will deal with the case in which $G_i=G$ later.	We describe how to construct~$G_{i+1}$ from $G_i$, so that Properties (P1)--(P4) are satisfied. In order to do that, we introduce the notion of $\updownarrow$-connected vertex and prove some lemmata about it. 

We say that a vertex $v$ of $G_i$ is \emph{$\updownarrow$-connected} if it satisfies at least one of the following properties:

\begin{itemize}
\item $v$ has a neighbor in $G$ above or on the top side of $\mathcal H_{i+1}$ and has a neighbor in $G$ below the bottom side of $\mathcal H_{i}$; 
\item $v$ has a neighbor in $G$ below or on the bottom side of $\mathcal H_{i+1}$ and has a neighbor in $G$ above the top side of $\mathcal H_{i}$. 
\end{itemize} 

The first lemma about $\updownarrow$-connected vertices states that if a vertex $u$ of $G_i$ that is an internal vertex of $G$ has a neighbor in $G$ above or on the top side of $\mathcal H_{i+1}$, then it also has a neighbor in $G$ below $\mathcal H_{i}$; roughly speaking, this is true because a ``very high'' neighbor of $u$ in $\Gamma$ pushes $u$ too high to be ``balanced'' (in terms of Equation~\ref{eq:y}) by neighbors of $u$ which lie in $\mathcal H_{i}$ or above.

\begin{lemma} \label{le:above-and-below-strip}
Let $u$ be a vertex of $G_i$ in $\mathcal I_G$. If $u$ has a neighbor in $G$ above or on the top side of $\mathcal H_{i+1}$, then it is $\updownarrow$-connected. Analogously, if $u$ has a neighbor in $G$ below or on the bottom side of $\mathcal H_{i+1}$, then it is $\updownarrow$-connected.
\end{lemma}

\begin{proof}
We prove the first part of the statement. The proof of the second part is analogous. Suppose, for a contradiction, that there exists a vertex $u$ of $G_i$ in $\mathcal I_G$ that has a neighbor $w$ in $G$ above or on the top side of $\mathcal H_{i+1}$ and that has no neighbor below $\mathcal H_{i}$. 

By Equation~\ref{eq:y}, we have that $\sum_{z\in \mathcal N_G(u)} (\lambda_{uz}\cdot y(z))=y(u)$. Since $\sum_{z\in \mathcal N_G(u)} \lambda_{uz} = 1$, it follows that $\sum_{z\in \mathcal N_G(u)} (\lambda_{uz}\cdot (y(z)-y(u)))=0$, hence 
\begin{equation}\label{eqn:strips}
	\lambda_{uw}\cdot (y(w)-y(u))=\sum_{z\in \mathcal N_G(u)\setminus \{w\}} (\lambda_{uz}\cdot(y(u)-y(z))).	
\end{equation}
Since the distance between the top side of $\mathcal H_{i+1}$ and the top side of $\mathcal H_i$ is equal to $\frac{h(i+1)-h(i)}{2}$, we have that 
\begin{eqnarray*}
	\lambda_{uw}\cdot (y(w)-y(u))\geq \lambda \cdot \frac{h(i+1)-h(i)}{2}.	
\end{eqnarray*}
Further, since every neighbor of $u$ in $G$ lies above or on the bottom side of $\mathcal H_i$ and since $u$ lies in the interior of $\mathcal H_i$, we have that $y(u)-y(z)<h(i)$, for every neighbor $z$ of $u$ in $G$. Hence, we have 
\begin{eqnarray*}
	\sum_{z\in \mathcal N_G(u)\setminus \{w\}} (\lambda_{uz}\cdot(y(u)-y(z)))<h(i) \cdot \sum_{z\in \mathcal N_G(u)\setminus \{w\}} \lambda_{uz} <h(i).	
\end{eqnarray*}
Furthermore, we have 
\begin{eqnarray*}
	\lambda\cdot \frac{h(i+1)-h(i)}{2}- h(i)>\lambda\cdot\frac{h(i+1)}{2}-3\cdot \frac{h(i)}{2}=
	\frac{\delta}{2} \left(\lambda \cdot \left(\frac{3}{\lambda}\right)^{i+1}-3\cdot\left(\frac{3}{\lambda}\right)^i\right)=0.	
\end{eqnarray*}
This implies that 
\begin{eqnarray*}
	\lambda_{uw}\cdot (y(w)-y(u))>\sum_{z\in \mathcal N_G(u)\setminus \{w\}} (\lambda_{uz}\cdot(y(u)-y(z))),	
\end{eqnarray*}
which contradicts Equation~\ref{eqn:strips}.
\end{proof}

The next lemma states that few vertices of $G_i$ are $\updownarrow$-connected (and thus, by Lemma~\ref{le:above-and-below-strip}, few vertices of~$G_i$ have neighbors in $G$ outside $\mathcal H_{i+1}$). 


\begin{lemma} \label{le:neighbors-above-below} 
The following statements hold true:

\begin{enumerate}[(S1)]
	\item $G_i$ contains at most two vertices that are in $\mathcal I_G$ and that are $\updownarrow$-connected; and
	\item if $G_i$ contains a vertex in $\mathcal O_G$, then it contains at most one vertex that is in $\mathcal I_G$ and that is $\updownarrow$-connected.
\end{enumerate}
\end{lemma}


\begin{proof}
For each $\updownarrow$-connected vertex $u$ of $G_i$, let $\ell(u)$ be the polygonal line defined as follows. Let $w$ be a neighbor of $u$ that lies above the top side of $\mathcal H_i$ and let $z$ be a neighbor of $u$ that lies below the bottom side of $\mathcal H_i$. Let $p_w$ and $p_z$ be the intersection points of the edges $(u,w)$ and $(u,z)$ with the top side of $\mathcal H_i$ and the bottom side of $\mathcal H_i$, respectively. Then $\ell(u)$ is composed of the line segments $\overline{up_w}$ and $\overline{up_z}$. For the remainder of the proof, refer to Figure~\ref{fig:few-updown-connected}. 

In order to prove (S1), suppose, for a contradiction, that $G_i$ contains (at least) three distinct $\updownarrow$-connected vertices $u_1$, $u_2$, and $u_3$ in $\mathcal I_G$; by the planarity of $\Gamma$, we have that $\ell(u_1)$, $\ell(u_2)$, and $\ell(u_3)$ do not cross each other; assume, w.l.o.g.\ up to a switch of the labels of $u_1$, $u_2$, and $u_3$, that $\ell(u_2)$ is in-between $\ell(u_1)$ and $\ell(u_3)$ in $\mathcal H_i$. Consider any path $P$ of $G_i$ connecting $u_1$ and $u_3$. By the planarity of $\Gamma$ and since $G_i$ is contained in the interior of $\mathcal H_i$, we have that $P$ does not cross $\ell(u_2)$, except at $u_2$. However, this implies that the removal of $u_2$ from $G_i$ separates $u_1$ from $u_3$, a contradiction to the fact that $G_i$ is biconnected. This proves (S1).

\begin{figure}[htb]\tabcolsep=4pt
	\centering
	\includegraphics[scale=0.7]{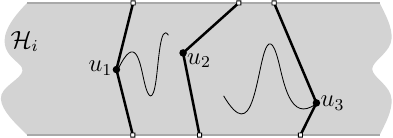}
	\caption{Illustration for the proof of Lemma~\ref{le:neighbors-above-below}. The fat lines are $\ell(u_1)$, $\ell(u_2)$, and $\ell(u_3)$.}
	\label{fig:few-updown-connected}
\end{figure}


In order to prove (S2), suppose, for a contradiction, that $G_i$ contains a vertex $u_1$ in $\mathcal O_G$ and contains (at least) two distinct $\updownarrow$-connected vertices $u_2$ and $u_3$ in $\mathcal I_G$. Since $u_2$ has neighbors both above the top side of $\mathcal H_i$ and below the bottom side of $\mathcal H_i$, the vertex $w$ of $G$ with the largest $y$-coordinate in $\Gamma$ lies above the top side of $\mathcal H_i$ and the vertex $z$ of $G$ with the smallest $y$-coordinate in $\Gamma$ lies below the bottom side of $\mathcal H_i$. Note that $(w,u_1,z)$ is a path incident to the outer face of $G$. As before, $\ell(u_1)$ is then defined as the portion of the polygonal line $(w,u_1,z)$ inside $\mathcal H_i$. By the planarity of $\Gamma$, we have that $\ell(u_1)$, $\ell(u_2)$, and $\ell(u_3)$ do not cross each other; further, since $u_1$ is incident to the outer face of $\Gamma$, we can assume, w.l.o.g.\ up to a switch of the labels of $u_2$ and $u_3$, that $\ell(u_2)$ is in-between $\ell(u_1)$ and $\ell(u_3)$ in $\mathcal H_i$. Consider any path $P$ of $G_i$ connecting $u_1$ and $u_3$. By the planarity of $\Gamma$ and since $G_i$ is contained in the interior of $\mathcal H_i$, we have that $P$ does not cross $\ell(u_2)$, except at $u_2$. However, this implies that the removal of $u_2$ from $G_i$ separates $u_1$ from $u_3$, a contradiction to the fact that $G_i$ is biconnected. This proves (S2) and~hence~the~lemma.
\end{proof}

Finally, we prove the following main lemma, stating that there exists a vertex that is an external vertex of $G_i$, that is an internal vertex of $G$, and that is not  $\updownarrow$-connected.

\begin{lemma} \label{le:at-least-one-not-up-down-connected}
There exists a vertex $u$ in $\mathcal O_{G_i}\cap \mathcal I_G$ that is not $\updownarrow$-connected.
\end{lemma}

\begin{proof}
Since $G_i$ contains at least four vertices (by Property~(P1) of $G_i$) and is biconnected, it follows that $|\mathcal O_{G_i}|\geq 3$. We distinguish three cases:
\begin{itemize}
	\item If $|\mathcal O_{G_i}\cap \mathcal O_{G}|=0$, that is, all the external vertices of $G_i$  are internal to $G$, then $|\mathcal O_{G_i}\cap \mathcal I_G|\geq 3$. By statement (S1) of Lemma~\ref{le:neighbors-above-below}, at least one of the vertices in $\mathcal O_{G_i}\cap \mathcal I_G$ is not $\updownarrow$-connected.   
	\item If $|\mathcal O_{G_i}\cap \mathcal O_{G}|=1$, that is, all but one of the external vertices of $G_i$ are internal to $G$, then $|\mathcal O_{G_i}\cap \mathcal I_G|\geq 2$. By statement (S2) of Lemma~\ref{le:neighbors-above-below}, at least one of the vertices in $\mathcal O_{G_i}\cap \mathcal I_G$ is not $\updownarrow$-connected.   
	\item If $|\mathcal O_{G_i}\cap \mathcal O_{G}|=2$, then $G_i$ contains two vertices $u$ and $w$ incident to the outer face of $G$ and $|\mathcal O_{G_i}\cap \mathcal I_G|\geq 1$; let $t$ be any vertex in $\mathcal O_{G_i}\cap \mathcal I_G$. Since the outer face of $G$ is delimited by a $3$-cycle $(u,w,z)$, we have that all the vertices of $G$ are contained in $\Gamma$ in the triangle with vertices $u$, $w$, and $z$.  By Property~(P4) of $G_i$, the vertices $u$ and $w$ are contained in the interior of $\mathcal H_i$, hence if $z$ lies above the bottom side of $\mathcal H_i$, then $t$ does not have any neighbor below the bottom line of $\mathcal H_i$, while if $z$ lies below the top side of $\mathcal H_i$, then $t$ does not have any neighbor above the top side of $\mathcal H_i$. In both cases, we have that $t$ is a vertex in $\mathcal O_{G_i}\cap \mathcal I_G$ that is not $\updownarrow$-connected.
\end{itemize}
The proof is concluded by observing that $|\mathcal O_{G_i}\cap \mathcal O_{G}|\leq 2$. Indeed, suppose, for a contradiction, that $|\mathcal O_{G_i}\cap \mathcal O_{G}|=3$. Then $G_i$ contains the three vertices incident to the outer face of $G$. Hence, by Property~(P2) of $G_i$, we have that $G_i$ contains the outer cycle of $G$, and by Property~(P3) of $G_i$, we have $G_i=G$, a contradiction. 
\end{proof}

\begin{figure}[htb]
\centering
\includegraphics[scale=0.85]{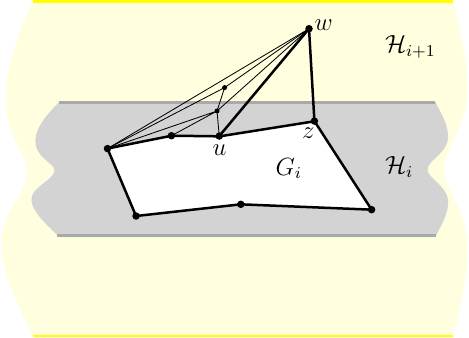}
\caption{Construction of $G_{i+1}$ from $G_i$.}
\label{fig:inductive-step}
\end{figure}

We are now ready to describe how to construct $G_{i+1}$ from $G_i$; refer to Figure~\ref{fig:inductive-step}. By Lemma~\ref{le:at-least-one-not-up-down-connected}, there exists a vertex $u$ in $\mathcal O_{G_i}\cap \mathcal I_G$ that is not $\updownarrow$-connected. Let $(u,z)$ be any of the two edges incident to $u$ in $\mathcal C_i$. The edge $(u,z)$ is incident to two faces of $G$, one inside and one outside $\mathcal C_i$; let $(u,z,w)$ be the cycle delimiting the face of $G$ incident to $(u,z)$ outside $\mathcal C_i$. We have the following.

\begin{lemma} \label{le:one-neighbor-inside}
The vertex $w$ does not belong to $G_i$. 
\end{lemma}

\begin{proof}
Suppose, for a contradiction, that $w$ belongs to $G_i$; this implies that $w$ belongs to $\mathcal C_i$, given that the face of $G$ delimited by $(u,z,w)$ lies outside $\mathcal C_i$. We distinguish two cases.

\begin{itemize}
	\item If any of the edges $(u,w)$ and $(w,z)$ does not belong to $\mathcal C_i$, then such an edge is an external chord of $\mathcal C_i$. However, this contradicts Property~(P2) of $G_i$. 
	\item If both the edges $(u,w)$ and $(w,z)$ belong to $\mathcal C_i$, then $\mathcal C_i$ is the cycle $(u,z,w)$, given that $(u,z)$ is an edge of $\mathcal C_i$ as well. However, since $(u,z,w)$ bounds a face $f$ of $G$ outside $\mathcal C_i$, we have that $f$ is the outer face of $G$ and thus $G_i=G$, a contradiction.
\end{itemize}

This concludes the proof. 
\end{proof}

By Lemma~\ref{le:one-neighbor-inside}, we have that $w$ does not belong to $G_i$. Let $G^w$ be the subgraph of $G$ induced by $\{w\}\cup V(\mathcal C_i)$. Note that $G^w$ is biconnected, as it contains $\mathcal C_i$ and the edges $(u,w)$ and $(w,z)$. Let $\mathcal C_{i+1}$ be the outer cycle of $G^w$ and let $G_{i+1}$ be the subgraph of $G$ inside $\mathcal C_{i+1}$. We prove that $G_{i+1}$ satisfies the required properties.	

\begin{itemize}
\item Property~(P3) is satisfied by construction.
\item Property~(P1) is satisfied, given that $G_i$ contains at least $i+3$ vertices (by Property~(P1) of $G_i$), and given that $G_{i+1}$ contains at least one more vertex than $G_i$, namely $w$.
\item  Property~(P2) is satisfied because $\mathcal C_{i+1}$ is the outer cycle of $G^w$ and $G^w$ is an induced subgraph of $G$.
\item  Property~(P4) is also satisfied by $\Gamma_{i+1}$. Namely, all the vertices of $G_i$ are contained inside $\mathcal H_i\subset \mathcal H_{i+1}$, by Property~(P4) of $G_i$. Further, $w$ is contained in the interior of $\mathcal H_{i+1}$, given that $u$ is not $\updownarrow$-connected and by Lemma~\ref{le:above-and-below-strip}. Moreover, all the vertices of $G_{i+1}$ are contained in the convex hull of $\{w\}\cup V(\mathcal C_i)$.
\end{itemize} 
Finally, $G_{i+1}$ is biconnected and internally-triangulated, by Property~\ref{pr:subgraph-inside-cycle}. 		

By Property~(P1), for some integer value $k\leq n-3$, we have that $G_k$ contains $n$ vertices, that is, $G_k=G$. By Property~(P4) of $G_k$, the $y$-extent of $\Gamma$ is smaller than $\delta \cdot\left(\frac{3}{\lambda}\right)^n$. Hence, the $y$-extent $Y$ of the triangle $\Delta$ delimiting the outer face of $\Gamma$ is smaller than $\delta  \cdot\left(\frac{3}{\lambda}\right)^n$. By Lemma~\ref{le:resolution-xy-triangle}, the $x$-extent $X$ of $\Delta$ is smaller than $\frac{\delta}{r}  \cdot\left(\frac{3}{\lambda}\right)^n$, where $r$ is the resolution of $\Delta$. The maximum distance between two separated geometric objects in $\Gamma$, which coincides with the longest side of $\Delta$, is smaller than $X+Y<\delta  \cdot\left(\frac{3}{\lambda}\right)^n\cdot\left(1+\frac{1}{r}\right)<\frac{2\delta}{r}\cdot\left(\frac{3}{\lambda}\right)^n$. 

The resolution of $\Gamma$ is larger than or equal to the ratio between the  minimum distance between two separated geometric objects in $\Gamma$, which is equal to $\delta$, and the upper bound on the  maximum distance between two separated geometric objects in $\Gamma$ obtained above. This ratio is $\frac{r}{2} \cdot \left(\frac{\lambda}{3}\right)^n$, which is indeed the bound in Theorem~\ref{th:lower-bound-gt-drawings}. By Property~\ref{pr:minimum-lambda}, we have that $\frac{r}{2} \cdot \left(\frac{\lambda}{3}\right)^n \in r\cdot \lambda^{O(n)}$.

\subsection{An Algebraic Approach to Lower Bound the Resolution of F-Drawings}\label{se:algebra}

We now discuss an approach to obtain a lower bound on the resolution of F-drawings that is different from the one presented in order to prove Theorem~\ref{th:lower-bound-gt-drawings}. This alternative approach is algebraic and is similar to the one proposed in~\cite{mrs-sge-11,rg-rsp-96} in order to obtain a lower bound on the resolution of T-drawings. Although we did not manage to bring the approach to completion, we deem it interesting and hence discuss it here.

Let $G$ be a $3$-connected\footnote{Note that the approach we describe here does not require the graph $G$ to be maximal.} plane graph with $n$ vertices. Let $v_1,\dots,v_N$ be the internal vertices of $G$ and let $v_{N+1},\dots,v_n$ be the external vertices of $G$. Also, let $\mathcal P$ be a convex polygon representing the outer cycle~$\mathcal C$ of~$G$, let $\Lambda$ be a coefficient matrix for $G$ with elements $\lambda_{v_iv_j}$, where $i=1,\dots,N$ and $j=1,\dots,n$, and let $\lambda$ be the smallest positive coefficient in $\Lambda$. Finally, let ${\bf x}$ and ${\bf y}$ be the vectors (with $N$ elements each) representing the coordinates of the internal vertices of $G$ in the F-drawing $\Gamma=(\Lambda,\mathcal P)$ of $G$. By Equation~\ref{eq:x}, the vector ${\bf x}$ is the (unique) solution to the following system of equations:

\begin{equation} \label{eq:matrix-form-algebraic}
\underbrace{
\begin{pmatrix}
		1 & -\lambda_{v_1v_2} & -\lambda_{v_1v_3} & \cdots & -\lambda_{v_1v_N}\\
		-\lambda_{v_2v_1} & 1 & -\lambda_{v_2v_3} & \cdots & -\lambda_{v_2v_N}\\
		-\lambda_{v_3v_1} & -\lambda_{v_3v_2} & 1 & \cdots & -\lambda_{v_3v_N}\\
		\cdots & \cdots & \cdots & \cdots & \cdots \\
		-\lambda_{v_Nv_1} & -\lambda_{v_Nv_2} & -\lambda_{v_Nv_3} & \cdots & 1
\end{pmatrix}}_{{\bf A}}
\cdot
\underbrace{\begin{pmatrix}
		x(v_1)\\
		x(v_2)\\
		x(v_3)\\
		\cdots\\
		x(v_N)
\end{pmatrix}}_{{\bf x}}
=
\underbrace{\begin{pmatrix}
		b_x(v_1)\\
		b_x(v_2)\\
		b_x(v_3)\\
		\cdots\\
		b_x(v_N)
\end{pmatrix}}_{{\bf b_x}}
\end{equation}

where 

\begin{equation}\label{eq:matrix-constants-algebraic}
{\bf b_x}=\begin{pmatrix}
	b_x(v_1)\\
	b_x(v_2)\\
	b_x(v_3)\\
	\cdots\\
	b_x(v_N)
\end{pmatrix}
= 
\begin{pmatrix}
	\lambda_{v_1v_{N+1}}\cdot x(v_{N+1})+\lambda_{v_1v_{N+2}}\cdot x(v_{N+2})+\cdots+\lambda_{v_1v_n}\cdot x(v_n)\\
	\lambda_{v_2v_{N+1}}\cdot x(v_{N+1})+\lambda_{v_2v_{N+2}}\cdot x(v_{N+2})+\cdots+\lambda_{v_2v_n}\cdot x(v_n)\\
	\lambda_{v_3v_{N+1}}\cdot x(v_{N+1})+\lambda_{v_3v_{N+2}}\cdot x(v_{N+2})+\cdots+\lambda_{v_3v_n}\cdot x(v_n)\\
	\cdots\\
	\lambda_{v_Nv_{N+1}}\cdot x(v_{N+1})+\lambda_{v_Nv_{N+2}}\cdot x(v_{N+2})+\cdots+\lambda_{v_Nv_n}\cdot x(v_n)
\end{pmatrix}.
\end{equation}

Analogously, by Equation~\ref{eq:y}, the vector ${\bf y}$ is the (unique) solution to the system of equations ${\bf A}\cdot {\bf y} = {\bf b_y}$, where ${\bf b_y}$ is a vector whose definition is analogous to Equation~\ref{eq:matrix-constants-algebraic}, however using $y$-coordinates rather than $x$-coordinates.

For a square matrix ${\bf M}$, denote by $\det({\bf M})$ its determinant. The above system of equations can be solved by Cramer's rule, obtaining, for $k=1,\dots,N$, that $x(v_k)$ is equal to $\det({\bf A^x_k})/\det({\bf A})$, where ${\bf A^x_k}$ denotes the matrix obtained by substituting the $k$-th column of ${\bf A}$ with ${\bf b_x}$. Note that both $\det({\bf A^x_k})$ and $\det({\bf A})$ are polynomial functions of degree at most $N$ whose variables are the elements $\lambda_{v_iv_j}$ of $\Lambda$. Similarly, for $k=1,\dots,N$, we have that $y(v_k)$ is equal to $\det({\bf A^y_k})/\det({\bf A})$, where ${\bf A^y_k}$ denotes the matrix obtained by substituting the $k$-th column of ${\bf A}$ with ${\bf b_y}$.

Let us now try to lower bound the distance $\delta_{ij}$ between two adjacent internal vertices $v_i$ and $v_j$ of $G$ in $\Gamma$. Note that the actual resolution of $\Gamma$ might be smaller than the distance between any two internal vertices of~$G$, as it is given by the distance between a vertex and an edge of $G$. However, trying to prove a lower bound for $\delta_{ij}$ is already sufficient to showcase the difficulty of finalizing this algebraic approach. The distance $\delta_{ij}$ is equal to $\sqrt{(x(v_i)-x(v_j))^2+(y(v_i)-y(v_j))^2}$. Using the above formulas we get:
\begin{eqnarray}\label{eq:distance}
\delta_{ij}=\frac{\sqrt{\left(\det({\bf A^x_i})-\det({\bf A^x_j})\right)^2+\left(\det({\bf A^y_i})-\det({\bf A^y_j})\right)^2}}{\det({\bf A})}=\frac{\sqrt{w_{ij}(\Lambda)}}{\det({\bf A})},
\end{eqnarray}

\noindent where $w_{ij}(\Lambda):=\left(\det({\bf A^x_i})-\det({\bf A^x_j})\right)^2+\left(\det({\bf A^y_i})-\det({\bf A^y_j})\right)^2$ is a polynomial of degree at most~$2N$ whose variables are the elements $\lambda_{v_iv_j}$ of~$\Lambda$. In order to prove a lower bound for $\delta_{ij}$, it suffices to prove a lower bound for $w_{ij}(\Lambda)$ and an upper bound for $\det({\bf A})$. The latter task is easy to accomplish, as each row of ${\bf A}$ is a vector with Euclidean length at most $\sqrt 2$ (since one term is $1$ and the absolute values of all other terms sum up to $1$), hence by Hadamard's inequality $\det({\bf A})$ is in $2^{O(N)}$. On the other hand, we do not know how to prove a lower bound for~$w_{ij}(\Lambda)$. Note that $w_{ij}(\Lambda)$ is the sum of $d(N)$ products, for some function~$d$, between $O(N)$ variables. However, even considering that each of such variables is larger than or equal to some fixed value $\lambda>0$, we do not know how to bound the sum away from $0$. Indeed, a polynomial might assume values arbitrarily close to zero even if the absolute value of each of its monomials is bounded away from zero by some (even arbitrarily large) value. 


\newcommand{\grado}{\textrm{deg}}

A similar approach was successfully employed in~\cite{mrs-sge-11,rg-rsp-96} in order to bound the resolution of a T-drawing of $G$ in which one can pick the shape of the polygon representing $\mathcal C$. We believe there are two main reasons for the different outcome of an algebraic approach in the two cases. The first reason is that in our setting, differently from~\cite{mrs-sge-11,rg-rsp-96}, we are not allowed to choose the polygon representing $\mathcal C$. Such a polygon is arbitrary and its resolution $r$ is a parameter for our problem. There is concrete hope to overcome this problem, at least for maximal plane graphs. Indeed, for a maximal plane graph $G$, the input triangle $\Delta$ representing the $3$-cycle $\mathcal C$ bounding the outer face can be obtained as an affine transformation of any (suitably chosen) triangle $\Delta'$. And even more, the F-drawing $(\Lambda,\Delta)$ can be obtained by the same affine transformation of the F-drawing $(\Lambda,\Delta')$, see Lemma~\ref{le:rotating}. Thus, one could try to derive the resolution of the drawing in which $\mathcal C$ is represented by $\Delta$ from the resolution of the drawing in which~$\mathcal C$ is represented by $\Delta'$, taking into account how the affine transformation modifies the latter resolution. The second, much more important, reason is that a T-drawing is an F-drawing in which each coefficient~$\lambda_{v_iv_j}$ is a fraction with numerator $1$ and denominator equal to the degree of $v_i$ in $G$, which we denote by $\grado(v_i)$. This allows one to suitably choose the coordinates of the polygon representing $\mathcal C$ so that the vertex coordinates are actually integers (and then distances are naturally bounded away from zero as a function of the maximum absolute value of a vertex coordinate). Indeed, consider the $x$-coordinates of a T-drawing of $G$, as the argument for the $y$-coordinates is the same. We can multiply the left and right sides of Equation~\ref{eq:matrix-form-algebraic} by~the~matrix
\begin{eqnarray*} 
\begin{pmatrix}
	\grado(v_1) & 0 & 0 & \cdots & 0\\
	0 & \grado(v_2) & 0 & \cdots & 0\\
	0 & 0 & \grado(v_3) & \cdots & 0\\
	\cdots & \cdots & \cdots & \cdots & \cdots \\
	0 & 0 & 0 & \cdots & \grado(v_N)
\end{pmatrix}.
\end{eqnarray*}
\noindent to transform Equation~\ref{eq:matrix-form-algebraic} into the following

\begin{equation} \label{eq:matrix-form-algebraic-revised}
\underbrace{
	\begin{pmatrix}
		\grado(v_1) & -a_{12} & -a_{13} & \cdots & -a_{1N}\\
		-a_{21} & \grado(v_2) & -a_{23} & \cdots & -a_{2N}\\
		-a_{31} & -a_{32} & \grado(v_3) & \cdots & -a_{3N}\\
		\cdots & \cdots & \cdots & \cdots & \cdots \\
		-a_{N1} & -a_{N2} & -a_{N3} & \cdots & \grado(v_N)\\
\end{pmatrix}}_{{\bf A'}}
\cdot
\underbrace{\begin{pmatrix}
		x(v_1)\\
		x(v_2)\\
		x(v_3)\\
		\cdots\\
		x(v_N)
\end{pmatrix}}_{{\bf x}}
=
\underbrace{\begin{pmatrix}
		\sum_{i=N+1}^n a_{1i}\cdot x(v_i)\\
		\sum_{i=N+1}^n a_{2i}\cdot x(v_i)\\
		\sum_{i=N+1}^n a_{3i}\cdot x(v_i)\\
		\cdots\\
		\sum_{i=N+1}^n a_{ni}\cdot x(v_i)\\
\end{pmatrix}}_{{\bf b'_x}},
\end{equation}

\noindent where $a_{ij}=1$ if $v_i$ and $v_j$ are neighbors in $G$, and $a_{ij}=0$ otherwise. Notice that all the elements of~${\bf A'}$ are integers. By using Cramer's rule, we get that $x(v_i)$ is equal to $\det({\bf A'_i})/\det({\bf A'})$, where ${\bf A'_i}$ is the matrix obtained by substituting the $i$-th column of ${\bf A'}$ with ${\bf b'_x}$. Thus, it suffices to choose the coordinates $x(v_{N+1}),x(v_{N+2}),\dots,x(v_{n})$ as integers multiple of $\det({\bf A'})$ in order to ensure that $x(v_i)$ is integer\footnote{The analyses of~\cite{mrs-sge-11,rg-rsp-96} have to use some extra care in the choice of the coordinates of the vertices of the polygon representing $\mathcal C$, as a suitable equilibrium of forces needs to be ensured on the vertices of $\mathcal C$, in order to ensure that the 2D drawing can then be lifted to a 3D convex polytope.}. The value of the determinant of ${\bf A'}$ can finally be bounded by a suitable exponential function of $n$. The reciprocal of such a function then provides an asymptotic lower bound on the resolution of the drawing.

\section{Upper Bound on the Resolution of F-Drawings}\label{se:upper-bound-gt-drawings}

In this section, we prove Theorem~\ref{th:upper-bound-gt-drawings}, which we restate here for the reader's convenience.

\medskip	
\rephrase{Theorem}{\ref{th:upper-bound-gt-drawings}}{%
There is a class of maximal plane graphs $\{G_n:n=5,6,\dots\}$, where $G_n$ has $n$ vertices, with the following property. For any $0<\lambda\leq \frac{1}{4}$ and $0<r\leq \frac{\sqrt 3}{2}$, there exist a triangle $\Delta$ with resolution~$r$ and a coefficient matrix $\Lambda$ for $G_n$ whose smallest positive coefficient is $\lambda$ such that the F-drawing $(\Lambda,\Delta)$ of $G_n$ has resolution in $r\cdot \lambda^{\Omega(n)}$. %
}
\medskip

We remark that, by Property~\ref{pr:minimum-lambda}, no coefficient matrix for $G_n$ can have a smallest coefficient larger than~$1/4$, given that $n\geq 5$. Furthermore, by Property~\ref{pro:resolution-triangle}, no triangle has resolution larger than $\sqrt 3/2$. 

The theorem is proved by analyzing a class of graphs introduced by Eades and Garvan~\cite{eg-dspg-95} and depicted in Figure~\ref{fig:upper-bound-static}. Consider any values $0<\lambda\leq \frac{1}{4}$ and $0<r\leq \frac{\sqrt 3}{2}$. 

\begin{figure}[htb]
\centering
\includegraphics[scale=1]{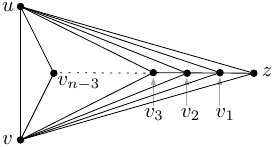}
\caption{The graph $G_n$ in the proof of Theorem~\ref{th:upper-bound-gt-drawings}. The $x$-coordinates of the vertices $z,v_1,\dots,v_{n-3}$ are larger than they should be, for the sake of readability.}
\label{fig:upper-bound-static}
\end{figure}
Let $\Delta$ be the triangle whose vertices have coordinates $p_u:=(0,0.5)$, $p_v:=(0,-0.5)$, and $p_z:=(r,0)$; note that the resolution of $\Delta$ is $r$. The vertices $u$, $v$, and $z$ of $G_n$ are embedded at $p_u$, $p_v$, and $p_z$, respectively. Further, let $v_0:=z$ and let $\Lambda$ be the coefficient matrix such that:
\begin{itemize}
\item $\lambda_{v_iv_{i+1}}=\lambda$, for $i=1,\dots,n-4$;
\item $\lambda_{v_{i+1}v_i}=\lambda$, for $i=0,\dots,n-4$;
\item $\lambda_{v_iu}=\lambda_{v_iv}=0.5-\lambda$, for $i=1,\dots,n-4$; and
\item $\lambda_{v_{n-3}u}=\lambda_{v_{n-3}v}=0.5-\lambda/2$.
\end{itemize}

Observe that, for every internal vertex $v_i$ of $G_n$, with $i\in \{1,\dots,n-3\}$, we have $\sum_{w\in \mathcal N(v_i)} \lambda_{v_iw}=1$. 

Easy calculations show that $y(v_i)=0$, for $i=1,\dots,n-3$; that is, all the vertices of $G_n$, except for $u$ and $v$, lie on the $x$-axis. Hence, the distance between any vertex $v_i$ and the edge $(u,v)$ is equal to $x(v_i)$.

By Equation~\ref{eq:x}, for $i=1,\dots,n-4$, we have
$$x(v_i)=\lambda_{v_iu}\cdot x(u) + \lambda_{v_iv}\cdot x(v) + \lambda_{v_iv_{i-1}}\cdot x(v_{i-1}) + \lambda_{v_iv_{i+1}}\cdot x(v_{i+1})=	 \lambda \cdot x(v_{i-1}) + \lambda \cdot x(v_{i+1}).    
$$ 

By the planarity of the F-drawing $(\Lambda,\Delta)$, we have $x(v_{i+1})<x(v_i)$, for $i=0,\dots,n-4$. Further, for $i=1,\dots,n-4$, by $x(v_i)=\lambda \cdot x(v_{i-1}) + \lambda \cdot x(v_{i+1})$ and  $x(v_{i+1})<x(v_i)$, we get $x(v_i)\leq \lambda \cdot x(v_{i-1}) + \lambda \cdot x(v_i)$, hence $x(v_i)\leq \frac{\lambda}{1-\lambda} \cdot x(v_{i-1})$. By repeatedly using the latter inequality, we get $x(v_{n-4})\leq x(v_0)\cdot \left(\frac{\lambda}{1-\lambda}\right)^{n-4}$. Since $x(v_0)=r$, we get $x(v_{n-4})\leq r\cdot\left(\frac{\lambda}{1-\lambda}\right)^{n-4}$. Since $\lambda\leq \frac{1}{4}=0.25$, we have that $\frac{\lambda}{1-\lambda}<\lambda^{\frac{1}{2}}$. Indeed, the last inequality is the same as $\lambda (\lambda^2-3\lambda+1)>0$, which is true for $0<\lambda<\frac{3-\sqrt 5}{2}\approx 0.38$. Thus, we get that $x(v_{n-4})\leq r\cdot \lambda^{\frac{n-4}{2}}\in r\cdot \lambda^{\Omega(n)}$. Hence, the distance between $v_{n-4}$ and $(u,v)$ is in $r\cdot \lambda^{\Omega(n)}$. Theorem~\ref{th:upper-bound-gt-drawings} then follows from the fact that the largest distance between any two separated geometric objects in the drawing is equal to $1$.

\section{Lower Bound on the Resolution of FG-Morphs}\label{se:lower-bound-morph}   

In this section, we prove Theorem~\ref{th:lower-bound-morph}, which we restate here for the reader's convenience.

\medskip	
\rephrase{Theorem}{\ref{th:lower-bound-morph}}{%
Let $\Gamma_0$ and $\Gamma_1$ be any two planar straight-line drawings of the same $n$-vertex maximal plane graph $G$ such that the outer faces of $\Gamma_0$ and $\Gamma_1$ are delimited by the same triangle $\Delta$.
There exists an FG-morph $\mathcal M=\{\Gamma_t: t\in[0,1]\}$ between $\Gamma_0$ and $\Gamma_1$ such that, for each $t\in[0,1]$, the resolution of $\Gamma_t$ is larger than or equal to $\left(r/n\right)^{O(n)}$, where $r$ is the minimum between the resolution of $\Gamma_0$ and $\Gamma_1$.
}
\medskip	

We first compute coefficient matrices $\Lambda_0$ and $\Lambda_1$  such that $\Gamma_0=(\Lambda_0,\Delta)$, such that $\Gamma_1=(\Lambda_1,\Delta)$, and such that the smallest positive coefficient in each of $\Lambda_0$ and $\Lambda_1$ is ``not too small''. This is a consequence of the following lemma.

\begin{lemma} \label{le:resolution-to-coefficient}
Let $\Gamma$ be a planar straight-line drawing of an $n$-vertex maximal plane graph $G$, let $\Delta$ be the triangle delimiting the outer face of $\Gamma$, and let $r$ be the resolution of $\Gamma$. There exists a coefficient matrix $\Lambda$ such that $\Gamma=(\Lambda,\Delta)$ and such that the smallest positive coefficient in $\Lambda$ is larger than $r/n$.
\end{lemma}

\begin{proof}
We employ and analyze a method proposed by Floater and Gotsman~\cite[Section 5]{fg-hmti-99}. This method is as follows. Refer to Figure~\ref{fig:coefficient-compute}.

\begin{figure}[htb]
	\centering
	\includegraphics[scale=1]{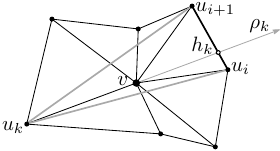}
	\caption{Illustration for the proof of Lemma~\ref{le:resolution-to-coefficient}.}
	\label{fig:coefficient-compute}
\end{figure}

Consider any internal vertex $v$ of $G$ and let $u_0,\dots,u_{d-1}$ be the clockwise order of the neighbors of $v$ in $G$. In the following description, we consider indices modulo $d$. Since $G$ is a maximal plane graph, $e_k:=(u_k,u_{k+1})$ is an edge of $G$ and $(v,u_k,u_{k+1})$ is a cycle bounding an internal face of $G$, for $k=0,\dots,d-1$. Consider each vertex $u_k$ independently. Shoot a ray $\rho_k$ starting at $u_k$ and passing through $v$; since the polygon representing the cycle $(u_0,u_1,\dots,u_{d-1})$ in $\Gamma$ is star-shaped, 
the ray $\rho_k$ hits either a vertex $u_i$ or the interior of an edge $e_i$. Since any point in the interior or on the boundary of a triangle can be expressed as a convex combination of the triangle's vertices (and the coefficients of such a convex combination are unique), we have $v=\mu_{k,k}\cdot u_k + \mu_{i,k}\cdot u_i + \mu_{i+1,k}\cdot u_{i+1}$, where $\mu_{k,k}>0$, $\mu_{i,k}>0$, $\mu_{i+1,k}\geq 0$, and  $\mu_{k,k}+\mu_{i,k}+\mu_{i+1,k}=1$; note that  $\mu_{i+1,k}= 0$ if $\rho_k$ passes through $u_i$. For every index $j\notin \{i,i+1,k\}$, set $\mu_{j,k}=0$. This concludes the work done when considering $u_k$. After the values $\mu_{j,k}$ have been computed for all $j,k\in \{0,\dots,d-1\}$, compute each coefficient $\lambda_{vu_k}$ as $\lambda_{vu_k}=\frac{1}{d}\sum_{j=0,\dots,d-1}\mu_{k,j}$. Note that $\lambda_{vu_k}>0$, for $k=0,\dots,d-1$, and that $\sum_{k=0,\dots,d-1}\lambda_{vu_k}=1$.

Let $\delta$ (let $D$) be smallest (resp.\ the largest) distance between two separated geometric objects in $\Gamma$, where $\delta/D=r$. In order to prove that each coefficient $\lambda_{vu_k}$ is larger than $r/n$, it suffices to prove that $\mu_{k,k}$ is larger than or equal to $r$; indeed, by construction, $\lambda_{vu_k}$ is larger than or equal to $\mu_{k,k}/d> \mu_{k,k}/n$. Consider again the triangle with vertices $u_k$, $u_i$, and $u_{i+1}$ that leads to the definition of $\mu_{k,k}$. Let $h_k$ be the intersection point between the ray $\rho_k$ and the edge $(u_i,u_{i+1})$.
We have that $\mu_{k,k}$ is equal to the ratio between $|\overline{u_kv}|$ and $|\overline{u_kh_k}|$. Furthermore, since $(u_k,v)$ is an edge of $G$, we have that $|\overline{u_kv}|\geq \delta$. Moreover, $|\overline{u_kh_k}|\leq D$, given that $|\overline{u_kh_k}|$ is smaller than or equal to the distance between the vertex $u_k$ and the edge  $(u_i,u_{i+1})$. Hence $\mu_{k,k}\geq \delta/D=r$ and $\lambda_{vu_k}>r/n$.
\end{proof}

The proof of Theorem~\ref{th:lower-bound-morph} proceeds as follows. Let $r_0$ be the resolution of $\Gamma_0$ and $r_1$ be the resolution of $\Gamma_1$; then $r=\min\{r_0,r_1\}$. First, by means of Lemma~\ref{le:resolution-to-coefficient}, we compute a coefficient matrix $\Lambda_0$ such that $\Gamma_0=(\Lambda_0,\Delta)$ and such that the smallest positive coefficient in $\Lambda_0$ is larger than $r_0/n\geq r/n$; further, again by Lemma~\ref{le:resolution-to-coefficient}, we compute a coefficient matrix $\Lambda_1$ such that $\Gamma_1=(\Lambda_1,\Delta)$ and such that the smallest positive coefficient in $\Lambda_1$ is larger than $r_1/n\geq r/n$. Let $\mathcal M=\{\Gamma_t=(\Lambda_t,\Delta): t\in[0,1]\}$ be the FG-morph between $\Gamma_0$ and $\Gamma_1$ such that, for any $t\in [0,1]$, $\Lambda_t=(1-t)\cdot \Lambda_0 + t \cdot \Lambda_1$. For any $t\in [0,1]$ and for any edge $(u,v)$ of $G$, where $u$ is an internal vertex of $G$, the coefficient $\lambda^t_{uv}$ in $\Lambda_t$ is equal to $(1-t)\cdot \lambda_{uv}^0+t\cdot \lambda_{uv}^1\geq (1-t)\cdot r/n+t\cdot r/n=r/n$. By Theorem~\ref{th:lower-bound-gt-drawings}, the resolution of the planar straight-line drawing $\Gamma_t$ is in $\left(r/n\right)^{O(n)}$. This concludes the proof of Theorem~\ref{th:lower-bound-morph}.

\section{Upper Bound on the Resolution of FG-Morphs}\label{se:upper-bound-morph}    

In this section, we prove Theorem~\ref{th:upper-bound-morphing}, which we restate here for the reader's convenience.

\medskip	
\rephrase{Theorem}{\ref{th:upper-bound-morphing}}{%
For every $n\geq 6$ multiple of $3$, there exist an $n$-vertex maximal plane graph $G$ and two planar straight-line drawings $\Gamma_0$ and $\Gamma_1$ of $G$ such that:  
\begin{enumerate}[(R1)]
	\item the outer faces of $\Gamma_0$ and $\Gamma_1$ are delimited by the same triangle $\Delta$; 
	\item the resolution of both $\Gamma_0$ and $\Gamma_1$ is larger than or equal to $c/n^2$, for some constant $c$; and
	\item any FG-morph between $\Gamma_0$ and $\Gamma_1$ contains a drawing whose resolution is in $1/2^{\Omega(n)}$.
\end{enumerate}
}
\medskip	

In order to prove the theorem, we employ a triangulated ``nested triangles graph'' (see, e.g.,~\cite{dlt-pepg-84,fp-nma-07}). Let $k=n/3$ and observe that $k$ is an integer. Then $G$ consists of (refer to Fig.~\ref{fig:upper-bound-drawings}): 
\begin{itemize}
\item $3$-cycles $(u_i,v_i,z_i)$, for $i=1,\dots,k$;
\item paths $(u_1,\dots,u_k)$, $(v_1,\dots,v_k)$, and $(z_1,\dots,z_k)$; and  
\item edges $(u_i,z_{i+1})$, $(z_i,v_{i+1})$, and $(v_i,u_{i+1})$, for $i=1,\dots,k-1$.
\end{itemize}
The outer cycle of $G$ is $(u_k,v_k,z_k)$.

\begin{figure}[htb]\tabcolsep=4pt
\centering
\begin{tabular}{c c}
	\includegraphics[scale=1]{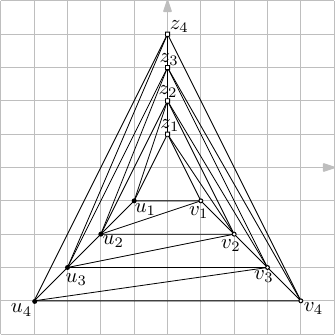}\hspace{5mm} & 		\includegraphics[scale=1]{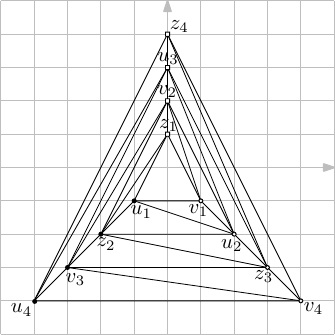}\\
	(a) \hspace{5mm} & (b) 
\end{tabular}

\caption{The graph $G$ in the proof of Theorem~\ref{th:upper-bound-morphing} (with $n=12$). (a) shows $\Gamma_0$ and (b) shows $\Gamma_1$.}
\label{fig:upper-bound-drawings}
\end{figure}

The planar straight-line drawings $\Gamma_0$ (refer to Fig.~\ref{fig:upper-bound-drawings}(a)) and $\Gamma_1$ (refer to Fig.~\ref{fig:upper-bound-drawings}(b)) both place the vertex set of $G$ at the point set composed of the points $(0,i)$, $(i,-i)$, and $(-i,-i)$, for $i=1,2,\dots,n/3$.

\noindent In particular, in $\Gamma_0$:
\begin{itemize}
\item the vertex $z_i$ lies at the point $(0,i)$, for $i=1,2,\dots,n/3$;
\item the vertex $v_i$ lies at the point $(i,-i)$, for $i=1,2,\dots,n/3$; and 
\item the vertex $u_i$ lies at the point $(-i,-i)$, for $i=1,2,\dots,n/3$.
\end{itemize}  
Further, in $\Gamma_1$:

\begin{itemize}
\item the vertex $z_i$ lies at the point $(0,i)$, for $i=k,k-3,\dots$, lies at the point $(i,-i)$, for $i=k-1,k-4,\dots$, and lies at the point $(-i,-i)$, for $i=k-2,k-5,\dots$;     
\item the vertex $v_i$ lies at the point $(i,-i)$, for $i=k,k-3,\dots$, lies at the point $(-i,-i)$, for $i=k-1,k-4,\dots$, and lies at the point $(0,i)$, for $i=k-2,k-5,\dots$; and
\item the vertex $u_i$ lies at the point $(-i,-i)$, for $i=k,k-3,\dots$, lies at the point $(0,i)$, for $i=k-1,k-4,\dots$, and lies at the point $(i,-i)$, for $i=k-2,k-5,\dots$.
\end{itemize}

The construction directly satisfies Property~(R1); indeed, the vertices $u_k$, $v_k$, and $z_k$ of the outer cycle of~$G$ are mapped to the points $(-n/3,-n/3)$, $(n/3,-n/3)$, and $(0,n/3)$, respectively, both in $\Gamma_0$ and~in~$\Gamma_1$.

We prove Property~(R2). For $i=0,1$, the drawing $\Gamma_i$ lies on an $O(n)\times O(n)$ grid, hence the largest distance between any two separated geometric objects in $\Gamma_i$ is in $O(n)$. Further, the smallest distance between any two separated geometric objects in $\Gamma_i$ is in $\Omega(1/n)$. Indeed, by Lemma~\ref{le:first-vertex-is-internal}, such a smallest distance is the distance between a vertex $v$ and an edge $e$ of $G$. By Pick's theorem, the area of the triangle~$T$ defined by $v$ and $e$ is at least $0.5$, given that $v$ and the end-vertices of $e$ lie at grid points; further, since the length of $e$ is in $O(n)$, it follows that the height of $T$ with respect to $e$, which coincides with the distance between $v$ and $e$,  is in $\Omega(1/n)$.

It remains to prove Property~(R3). To do that, we first prove a lower bound on some coefficients of any coefficient matrices $\Lambda_0$ and $\Lambda_1$ such that $\Gamma_0=(\Lambda_0,\Delta)$ and $\Gamma_1=(\Lambda_1,\Delta)$. Recall that $\lambda^t_{vu}$ denotes the element of a coefficient matrix $\Lambda_t$ whose row corresponds to a vertex $v$ and whose column corresponds to a vertex $v$.

\begin{claimX} \label{cl:upper-bound-coefficients}
For any coefficient matrix $\Lambda_0$ such that $\Gamma_0=(\Lambda_0,\Delta)$ and for any $i=2,\dots, k-1$, we have $\lambda_{u_iu_{i+1}}^0>0.5$, $\lambda_{v_iv_{i+1}}^0>0.5$, and $\lambda_{z_iz_{i+1}}^0>0.5$. Analogously, for any coefficient matrix $\Lambda_1$ such that $\Gamma_1=(\Lambda_1,\Delta)$ and for any $i=2,\dots, k-1$, we have $\lambda_{u_iz_{i+1}}^1>0.5$, $\lambda_{v_iu_{i+1}}^1>0.5$, and $\lambda_{z_iv_{i+1}}^1>0.5$.
\end{claimX}

\begin{proof}
Consider any coefficient matrix $\Lambda_0$ such that $\Gamma_0=(\Lambda_0,\Delta)$. We prove that $\lambda_{z_iz_{i+1}}^0>0.5$. By Equation~\ref{eq:y}, we have that $y(z_i)=\sum_{w\in \mathcal N(z_i)} \lambda_{z_iw}^0\cdot y(w)$, where $\mathcal N(z_i)=\{u_i,v_i,z_{i-1},z_{i+1},u_{i-1},v_{i+1}\}$. Since the coefficients $\lambda_{z_iw}^0$ with $w\in \mathcal N(z_i)$ are all positive and since the values $y(u_i)$, $y(v_i)$, $y(u_{i-1})$, and $y(v_{i+1})$ are all smaller than $y(z_{i-1})=i-1$, we get 

$$y(z_i)=i<(\lambda_{z_iu_i}^0+\lambda_{z_iv_i}^0+\lambda_{z_iu_{i-1}}^0+\lambda_{z_iv_{i+1}}^0+\lambda_{z_iz_{i-1}}^0)\cdot (i-1) + \lambda_{z_iz_{i+1}}^0\cdot (i+1).$$
From this, it follows that $$\lambda_{z_iz_{i+1}}^0>\lambda_{z_iu_i}^0+\lambda_{z_iv_i}^0+\lambda_{z_iu_{i-1}}^0+\lambda_{z_iv_{i+1}}^0+\lambda_{z_iz_{i-1}}^0,$$ which gives us $\lambda_{z_iz_{i+1}}^0>0.5$, given that  $\sum_{w\in \mathcal N(z_i)} \lambda_{z_iw}^0=1$.

The proof that $\lambda_{v_iv_{i+1}}^0>0.5$ and $\lambda_{z_iz_{i+1}}^0>0.5$ uses very similar arguments, however by considering the $x$-coordinates, rather than the $y$-coordinates, and by employing Equation~\ref{eq:x} in place of Equation~\ref{eq:y}. The proof that $\lambda_{u_iz_{i+1}}^1>0.5$, $\lambda_{v_iu_{i+1}}^1>0.5$, and $\lambda_{z_iv_{i+1}}^1>0.5$ for any coefficient matrix $\Lambda_1$ such that $\Gamma_1=(\Lambda_1,\Delta)$ again uses very similar arguments and is hence omitted.
\end{proof}

Consider now any coefficient matrices $\Lambda_0$ and $\Lambda_1$ such that $\Gamma_0=(\Lambda_0,\Delta)$ and $\Gamma_1=(\Lambda_1,\Delta)$, and consider the corresponding FG-morph $\mathcal M=\{\Gamma_t=(\Lambda_t,\Delta): t\in[0,1]\}$. We are going to prove that the resolution of the ``intermediate'' drawing $\Gamma_{0.5}$ of $\mathcal M$ is exponentially small. By Claim~\ref{cl:upper-bound-coefficients}, we have $\lambda_{u_iu_{i+1}}^0>0.5$ and $\lambda_{u_iz_{i+1}}^1>0.5$. This, together with the fact that $\lambda_{u_iu_{i+1}}^{0.5}=0.5\cdot \lambda_{u_iu_{i+1}}^0 + 0.5 \cdot \lambda_{u_iu_{i+1}}^1$ and $\lambda_{u_iz_{i+1}}^{0.5}=0.5\cdot \lambda_{u_iz_{i+1}}^0 + 0.5 \cdot \lambda_{u_iz_{i+1}}^1$, implies that $\lambda_{u_iu_{i+1}}^{0.5}>0.25$ and $\lambda_{u_iz_{i+1}}^{0.5}>0.25$. Analogously, we have $\lambda_{v_iv_{i+1}}^{0.5}>0.25$, $\lambda_{v_iu_{i+1}}^{0.5}>0.25$, $\lambda_{z_iz_{i+1}}^{0.5}>0.25$, and $\lambda_{z_iv_{i+1}}^{0.5}>0.25$. Roughly speaking, this means that, in $\Gamma_{0.5}$, the vertices of the $3$-cycle $(u_i,v_i,z_i)$ receive a non-negligible ``attraction'' from all their neighbors in the $3$-cycle $(u_{i+1},v_{i+1},z_{i+1})$. Hence, each of $u_i$, $v_i$, and $z_i$ is ``not too close'' to any of its neighbors in $(u_{i+1},v_{i+1},z_{i+1})$, which implies that a constant fraction of the area of the triangle $(u_{i+1},v_{i+1},z_{i+1})$ is external to the triangle $(u_i,v_i,z_i)$. Since this is true for every $i=2,\dots,k-1$, the exponential bound in (R3) follows. We now make this argument precise. Denote by $\mathcal A(T)$ the area of a triangle $T$. Refer to Figure~\ref{fig:LB-morphing-proof}.

\begin{figure}[htb]
\centering
\includegraphics[scale=1]{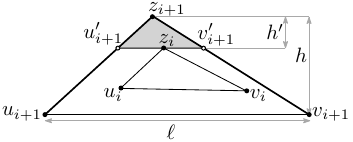}
\caption{Illustration for the proof that the area of $\Delta_i$ is a fraction of the area of $\Delta_{i+1}$. The gray triangle is $\Delta'_{i+1}$.}
\label{fig:LB-morphing-proof}
\end{figure}

Let $\Delta_i$ and $\Delta_{i+1}$ be the representations of the $3$-cycles $(u_i,v_i,z_i)$ and $(u_{i+1},v_{i+1},z_{i+1})$, respectively, in $\Gamma_{0.5}$. Assume that the longest side $s$ of $\Delta_{i+1}$ is the one connecting $u_{i+1}$ with $v_{i+1}$, as the other cases are analogous. Let $\ell$ be the length of $s$ and let $h$ be the height of $\Delta_{i+1}$ with respect to $s$. By Lemma~\ref{le:rotating}, we can assume that the $x$-axis passes through $s$ and that $z_{i+1}$ lies above $s$. 

Now consider the vertex $t$ of $\Delta_i$ with the highest $y$-coordinate. We have that $t\neq v_i$; indeed, if $t$ were equal to $v_i$, then it would be possible to add the edge $(v_i,z_{i+1})$ to $\Gamma_{0.5}$ without crossing any edge of $\Gamma_{0.5}$, contradicting the fact that $G$ is a maximal plane graph. It follows that $t$ is either $z_i$ or $u_i$. We only discuss the case in which $t=z_i$, as the other case is analogous. 

We are going to bound the $y$-coordinate of $z_i$ in terms of the $y$-coordinate of $z_{i+1}$. By Equation~\ref{eq:y}, in~$\Gamma_{0.5}$ we have $y(z_i)=\sum_{w\in \mathcal N(z_i)} \lambda_{z_iw}^{0.5}\cdot y(w)$, where $\mathcal N(z_i)=\{u_i,v_i,z_{i-1},z_{i+1},u_{i-1},v_{i+1}\}$. By the assumption that $z_i$ is the highest vertex of $\Delta_i$ and by the planarity of $\Gamma_{0.5}$, we have that every neighbor of~$z_i$ different from $z_{i+1}$ lies below $z_i$ (possibly $u_i$ lies on the same horizontal line as $z_i$). This implies~that $$y(z_i)<(\lambda_{z_iu_i}^{0.5}+\lambda_{z_iv_i}^{0.5}+\lambda_{z_iz_{i-1}}^{0.5}+\lambda_{z_iu_{i-1}}^{0.5})\cdot y(z_i) + \lambda_{z_iz_{i+1}}^{0.5}\cdot y(z_{i+1})+\lambda_{z_iv_{i+1}}^{0.5}\cdot y(v_{i+1}).$$ Since $y(v_{i+1})=0$, this is $$y(z_i)<\frac{\lambda_{z_iz_{i+1}}^{0.5}}{1-(\lambda_{z_iu_i}^{0.5}+\lambda_{z_iv_i}^{0.5}+\lambda_{z_iz_{i-1}}^{0.5}+\lambda_{z_iu_{i-1}}^{0.5})}\cdot y(z_{i+1}).$$
Since $\lambda_{z_iv_{i+1}}^{0.5}>0.25$ and $\lambda_{z_iz_{i+1}}^{0.5}>0.25$, we have $$\lambda_{z_iz_{i+1}}^{0.5}+\lambda_{z_iu_i}^{0.5}+\lambda_{z_iv_i}^{0.5}+\lambda_{z_iz_{i-1}}^{0.5}+\lambda_{z_iu_{i-1}}^{0.5}<0.75$$ and $$\lambda_{z_iu_i}^{0.5}+\lambda_{z_iv_i}^{0.5}+\lambda_{z_iz_{i-1}}^{0.5}+\lambda_{z_iu_{i-1}}^{0.5}<0.5.$$ Hence, by setting $$\rho := \lambda_{z_iu_i}^{0.5}+\lambda_{z_iv_i}^{0.5}+\lambda_{z_iz_{i-1}}^{0.5}+\lambda_{z_iu_{i-1}}^{0.5},$$ we get that $y(z_i)<\frac{0.75-\rho}{1-\rho}\cdot y(z_{i+1})$, where $\rho\in (0,0.5)$. As the function $f(\rho):=\frac{0.75-\rho}{1-\rho}$ decreases as $\rho$ increases over the interval $(0,0.5)$, we get that $\frac{0.75-\rho}{1-\rho}<0.75$ and hence $y(z_i)<0.75\cdot y(z_{i+1})$ in $\Gamma_{0.5}$.

Consider now a horizontal line through $z_{i}$ and let $u'_{i+1}$ and $v'_{i+1}$ be its intersection points with the edges $\overline{z_{i+1}u_{i+1}}$ and $\overline{z_{i+1}v_{i+1}}$, respectively. Let $\Delta'_{i+1}$ be the triangle with vertices $z_{i+1}$, $u'_{i+1}$, and $v'_{i+1}$. The height $h'$ of $\Delta'_{i+1}$ with respect to the side $\overline{u_{i+1}v_{i+1}}$ is larger than $0.25\cdot h$, given that $h'$ coincides with $y(z_{i+1})-y(z_i)>0.25\cdot y(z_{i+1})$, and given that $h=y(z_{i+1})$. By the similarity of $\Delta_{i+1}$ and $\Delta'_{i+1}$, it follows that the length $\ell'$ of the side $\overline{u'_{i+1}v'_{i+1}}$ of $\Delta'_{i+1}$ is larger than $0.25 \cdot \ell$. Hence, $\mathcal A(\Delta'_{i+1})=h'\cdot\ell'/2\geq (0.25)^2 \cdot \mathcal A(\Delta_{i+1})=0.0625 \cdot \mathcal A(\Delta_{i+1})$. 

Since $\Delta'_{i+1}$ lies entirely above or on the horizontal line through $z_i$, while $\Delta_i$ lies entirely below or on the line through $z_i$, it follows that the interiors of $\Delta_i$ and $\Delta'_{i+1}$ are disjoint, and hence $\mathcal A(\Delta_i)\leq \mathcal A(\Delta_{i+1})-\mathcal A(\Delta'_{i+1})\leq (1-0.0625)\cdot \mathcal A(\Delta_{i+1})=0.9375\cdot \mathcal A(\Delta_{i+1})$. This is what we aimed to proved: the area of $\Delta_i$ is a constant fraction of the area of $\Delta_{i+1}$. Since this holds for $i=2,\dots,k-1$, we have that the ratio between $\mathcal A(\Delta_{2})$ and $\mathcal A(\Delta_{k})$ is in $1/2^{\Omega(n)}$. Since the shortest height, and hence the shortest distance $\delta$ between two separated geometric objects, in $\Delta_{2}$ is at most $c\sqrt {\mathcal A(\Delta_{2})}$, for some constant $c$, while the longest side $D$ of $\Delta_{k}$ is at least $c'\sqrt {\mathcal A(\Delta_{k})}$, for some constant $c'$, the ratio $\delta/D$ is in $1/2^{\Omega(n)}$. This concludes the proof of Property (R3) and of Theorem~\ref{th:upper-bound-morphing}.

\section{Piecewise Linear FG-Morphs} \label{se:approximating-FG}

In this section, we prove that our results on the resolution of FG-morphs can be applied to ``approximate'' such morphs with a finite number of linear morphs. 

A \emph{linear morph} $\langle \Gamma_0, \Gamma_1 \rangle$ between two planar straight-line drawings $\Gamma_0$ and $\Gamma_1$ of the same plane graph moves each vertex with uniform speed along a straight-line segment from its position in $\Gamma_0$ to its position in $\Gamma_1$; a linear morph does not necessarily guarantee planarity for its intermediate drawings, see, e.g.,~\cite{aaccdd-hmpgd-17,ekp-ifm-03}. A \emph{piecewise linear morph} between two planar straight-line drawings $\Gamma_0$ and $\Gamma_1$ of the same plane graph is a morph between $\Gamma_0$ and $\Gamma_1$ composed of a sequence of linear morphs. Hence, a piecewise linear morph can be described by a sequence $\langle \Gamma_0=\Psi_0,\Psi_1,\dots,\Psi_k=\Gamma_1 \rangle$ of planar straight-line drawings, where $\langle \Psi_i, \Psi_{i+1} \rangle$ is a linear morph, for $i=0,\dots,k-1$. Piecewise linear morphs have been thoroughly investigated in recent years, see, e.g.,~\cite{aaccdd-hmpgd-17,adflpr-omcd-15,abc-pd-19,bbd-hmt-19,bhl-msd-19,befklow-mpgd3d-23,ddf-upm-20,DBLP:journals/jgaa/EricksonL23,DBLP:journals/comgeo/KleistKLSSS19}; notably, it is known that, for any two planar straight-line drawings of the same $n$-vertex plane graph, there exists a piecewise linear morph consisting of $O(n)$ linear morphs that preserves the planarity of the drawing at any time~\cite{aaccdd-hmpgd-17,DBLP:journals/jgaa/EricksonL23}. 	

We show how Theorem~\ref{th:lower-bound-morph} implies that, for any two drawings $\Gamma_0$ and $\Gamma_1$ of a plane graph, there exist an FG-morph $\mathcal M$ between $\Gamma_0$ and $\Gamma_1$ and a finite sequence of drawings in $\mathcal M$, such that the first drawing in the sequence is $\Gamma_0$, the last drawing in the sequence is $\Gamma_1$, and the piecewise linear morph defined by such a sequence of drawings is planar. Note that the piecewise linear morph is not part of $\mathcal M$, whereas it is a sequence of linear morphs, where each linear morph occurs between two drawings of $\mathcal M$.  

\begin{theorem} \label{th:linearize}
Let $\Gamma_0$ and $\Gamma_1$ be any two planar straight-line drawings of an $n$-vertex maximal plane graph $G$ such that the outer faces of $\Gamma_0$ and $\Gamma_1$ are delimited by the same triangle $\Delta$. Let $r$ be the minimum between the resolution of $\Gamma_0$ and $\Gamma_1$. There exist an FG-morph $\mathcal M=\{\Gamma_t: t\in [0,1]\}$ between $\Gamma_0$ and $\Gamma_1$ and a sequence $\langle \Psi_0,\Psi_1,\dots,\Psi_k\rangle$ of planar straight-line drawings of $G$ such that: 

\begin{enumerate}[(a)]
	\item $\Psi_0=\Gamma_0$ and $\Psi_k=\Gamma_1$;
	\item for $i=0,1,\dots,k$, there is a value $t_i\in [0,1]$ such that $\Psi_i=\Gamma_{t_i}$, i.e., $\Psi_i$ is a drawing in the morph~$\mathcal M$;
	\item for $i=0,1,\dots,k-1$, we have $t_i<t_{i+1}$;
	\item for $i=0,1,\dots,k-1$, the linear morph $\langle \Psi_i,\Psi_{i+1}\rangle$ is planar; and
	\item $k\in \left(n/r\right)^{O(n)}$.
\end{enumerate}
\end{theorem}

It is intuitive that a piecewise linear morph $\langle \Psi_0,\Psi_1,\dots,\Psi_k\rangle$ satisfying Properties~(a)--(d) exists (roughly speaking, this can be obtained by mimicking the FG-morph $\mathcal M$ by a sequence of ``suitably short'' linear morphs); it is however not obvious, in our opinion, that one of such piecewise linear morphs exists composed only of a finite number of linear morphs, as ensured by Property~(e).

The remainder of this section contains a proof of Theorem~\ref{th:linearize}. First, we let $\mathcal M=\{\Gamma_t: t\in [0,1]\}$ be the morph in Theorem~\ref{th:lower-bound-morph}, which ensures that, for every $t\in [0,1]$, the  resolution of the drawing $\Gamma_t$ is larger than or equal to $\left(r/n\right)^{c\cdot n}$, for some constant $c$. Let $N$ be the number of internal vertices of $G$ (hence $N=n-3$). Let $v_1,\dots,v_N$ be an arbitrary order of the internal vertices of $G$; further, let $v_{n-2},v_{n-1},v_n$ be an arbitrary order of the external vertices of $G$. For $i=1,\dots,n$, denote by $x^t(v_i)$ and $y^t(v_i)$ the coordinates of $v_i$ in $\Gamma_t$. Finally, let $D$ be the length of the longest side of $\Delta$. 

We introduce the \emph{FG-curve} $\mathcal C_{\mathcal M}$ for $\mathcal M$. Let $z_1,z_2,\dots,z_{2N+1}$ denote the coordinates of the Euclidean space $\mathbb{R}^{2N+1}$. The curve~$\mathcal C_{\mathcal M}$, which lies in $\mathbb{R}^{2N+1}$, is defined as follows: For every $t\in [0,1]$, the curve~$\mathcal C_{\mathcal M}$ contains a point $p^t$ whose first coordinate is $t$ and whose $(2i)$-th and $(2i+1)$-th coordinates are the $x$- and $y$-coordinates of $v_i$ in $\Gamma_t$, respectively, for $i=1,\dots,N$. That is, for every $t\in [0,1]$:
\begin{itemize}
\item $z_1(p^t)=t$;
\item $z_{2i}(p^t)=x^t(v_i)$, for $i=1,\dots,N$; and
\item $z_{2i+1}(p^t)=y^t(v_i)$, for $i=1,\dots,N$.
\end{itemize}

The first ingredient in the proof of Theorem~\ref{th:linearize} is an upper bound on the length $\ell_{\mathcal M}$ of $\mathcal C_{\mathcal M}$. 

\begin{claimX}  \label{cl:curve-length}
The length $\ell_{\mathcal M}$ of $\mathcal C_{\mathcal M}$ is in $O(D \cdot N^3)$.
\end{claimX}

\begin{proof} 
A high-level view of the proof of the claim is as follows. First, we bound the number of local minima and maxima of the curve $\mathcal C_{\mathcal M}$ with respect to each coordinate $z_l$; this is done algebraically. Then we use such local minima and maxima to split $\mathcal C_{\mathcal M}$ into monotone curves, for which it is easier to upper bound the length. The upper bound on the length of  $\mathcal C_{\mathcal M}$ is finally obtained as the sum of the upper bounds on the lengths of the individual monotone curves.

We now provide the details of such a proof. Consider any internal vertex $v_i$ of $G$ and any $t\in [0,1]$. By Equation~\ref{eq:x}, we have $x^t(v_i)=\sum_{v_j\in \mathcal N(v_i)} (\lambda^t_{v_iv_j}\cdot x^t(v_j))$, where $\lambda^t_{v_iv_j}=(1-t)\cdot \lambda^0_{v_iv_j} + t\cdot \lambda^1_{v_iv_j}$; observe that the values $\lambda^0_{v_iv_j}$ and $\lambda^1_{v_iv_j}$ are fixed in the definition of $\mathcal M$, as in the proof of Theorem~\ref{th:lower-bound-morph}. With a slight overload of notation, we let $\lambda^t_{v_iv_j}:=0$ for every vertex $v_j$ of $G$ which is not a neighbor of~$v_i$. Then Equation~\ref{eq:x} over all the internal vertices of $G$ can be expressed in matrix form ${\bf A^t} \cdot {\bf x^t}={\bf b^t_x}$ as~follows:
\medskip

\begin{equation} \label{eq:matrix-form}
	\underbrace{
		\begin{pmatrix}
			1 & -\lambda^t_{v_1v_2} & -\lambda^t_{v_1v_3} & \cdots & -\lambda^t_{v_1v_N}\\
			-\lambda^t_{v_2v_1} & 1 & -\lambda^t_{v_2v_3} & \cdots & -\lambda^t_{v_2v_N}\\
			-\lambda^t_{v_3v_1} & -\lambda^t_{v_3v_2} & 1 & \cdots & -\lambda^t_{v_3v_N}\\
			\cdots & \cdots & \cdots & \cdots & \cdots \\
			-\lambda^t_{v_Nv_1} & -\lambda^t_{v_Nv_2} & -\lambda^t_{v_Nv_3} & \cdots & 1
	\end{pmatrix}}_{{\bf A^t}}
	\cdot
	\underbrace{\begin{pmatrix}
			x^t(v_1)\\
			x^t(v_2)\\
			x^t(v_3)\\
			\cdots\\
			x^t(v_N)
	\end{pmatrix}}_{{\bf x^t}}
	=
	\underbrace{\begin{pmatrix}
			b^t_x(v_1)\\
			b^t_x(v_2)\\
			b^t_x(v_3)\\
			\cdots\\
			b^t_x(v_N)
	\end{pmatrix}}_{{\bf b^t_x}}
\end{equation}

where 

\begin{equation}\label{eq:matrix-constants}
	{\bf b^t_x}=\begin{pmatrix}
		b^t_x(v_1)\\
		b^t_x(v_2)\\
		b^t_x(v_3)\\
		\cdots\\
		b^t_x(v_N)
	\end{pmatrix}
	= 
	\begin{pmatrix}
		\lambda^t_{v_1v_{n-2}}\cdot x^t(v_{n-2})+\lambda^t_{v_1v_{n-1}}\cdot x^t(v_{n-1})+\lambda^t_{v_1v_n}\cdot x^t(v_n)\\
		\lambda^t_{v_2v_{n-2}}\cdot x^t(v_{n-2})+\lambda^t_{v_2v_{n-1}}\cdot x^t(v_{n-1})+\lambda^t_{v_2v_n}\cdot x^t(v_n)\\
		\lambda^t_{v_3v_{n-2}}\cdot x^t(v_{n-2})+\lambda^t_{v_3v_{n-1}}\cdot x^t(v_{n-1})+\lambda^t_{v_3v_n}\cdot x^t(v_n)\\
		\cdots\\
		\lambda^t_{v_Nv_{n-2}}\cdot x^t(v_{n-2})+\lambda^t_{v_Nv_{n-1}}\cdot x^t(v_{n-1})+\lambda^t_{v_Nv_n}\cdot x^t(v_n)
	\end{pmatrix}.
\end{equation}

Observe that the $x$-coordinates $x^t(v_{n-2})$, $x^t(v_{n-1})$, and $x^t(v_{n})$ in Equation~\ref{eq:matrix-constants} are the ones of the external vertices of $G$, hence they have fixed values (independent of $t$). Thus, each term $b^t_x(v_i)$ is just a linear function of $t$. Equation~\ref{eq:y} over all the internal vertices of $G$ can be analogously expressed in matrix form ${\bf A^t} \cdot {\bf y^t}={\bf b^t_y}$; observe that the matrix ${\bf A^t}$ is the same as in Equation~\ref{eq:matrix-form}.

Using Cramer's rule, we can obtain the value of $z_{2i}(p^t)=x^t(v_i)$ as the ratio $\mathcal N^t_{x,i}/\mathcal D^t$ of two determinants $\mathcal N^t_{x,i}$ and $\mathcal D^t$. The determinant $\mathcal D^t$ is just $\det({\bf A^t})$, while $\mathcal N^t_{x,i}$ is the determinant of the matrix obtained from ${\bf A^t}$ by substituting the $i$-th column with ${\bf b^t_x}$. Since every element of the matrix ${\bf A^t}$ and of the vector ${\bf b^t_x}$ is a linear function of $t$, since ${\bf A^t}$ and ${\bf b^t_x}$ have $N$ rows, and by Laplace expansion, we have that both $\mathcal N^t_{x,i}$ and $\mathcal D^t$ are polynomial functions of $t$ with degree at most $N$. 

This implies that the number of local minima and maxima of the curve $\mathcal C_{\mathcal M}$ with respect to the coordinate $z_{2i}$ is in $O(N^2)$, for each $i=1,\dots,N$. Namely, this number coincides with the number of solutions of the equation $\frac{\partial(\mathcal N^t_{x,i}/\mathcal D^t)}{\partial t}=0$; this, in turn, coincides with the number of solutions of the equation $\frac{\partial(\mathcal N^t_{x,i})}{\partial t}\cdot \mathcal D^t -\mathcal N^t_{x,i}\cdot\frac{\partial(\mathcal D^t)}{\partial t}=0$. Since the left term of the last equation is a polynomial of degree at most $N\cdot(N-1)$, the number of solutions to the equation is indeed in $O(N^2)$. Analogously, the number of local minima and maxima of the curve $\mathcal C_{\mathcal M}$ with respect to the coordinate $z_{2i+1}$ is in $O(N^2)$, for each $i=1,\dots,N$. Let $t_0=0<t_1<t_2<\dots<t_m=1$ be the values of $t$ for which the curve $\mathcal C_{\mathcal M}$ achieves a minimum or a maximum with respect to some coordinate.

For $j=1,\dots,m$, denote by $\mathcal C_{\mathcal M,j}$ the part of $\mathcal C_{\mathcal M}$ within the interval $[t_{j-1},t_j]$. Then $\mathcal C_{\mathcal M,j}$ is \emph{monotone}; that is, each coordinate $z_l$ monotonically increases or monotonically decreases over the interval $[t_{j-1},t_j]$. For $j=1,\dots,m$, let $\ell_{\mathcal M,j}$ be the length of $\mathcal C_{\mathcal M,j}$. Note that $\mathcal C_{\mathcal M}=\mathcal C_{\mathcal M,1}\cup \mathcal C_{\mathcal M,2}\cup  \dots\cup \mathcal C_{\mathcal M,m}$ and $\ell_{\mathcal M}=\sum_{j=1}^m \ell_{\mathcal M,j}$. Then the length of $\mathcal C_{\mathcal M,j}$ is upper bounded by the sum of the $z_l$-extents of the curves $\mathcal C_{\mathcal M,j}$ (over all coordinates $z_l$). Recall that $\mathcal C_{\mathcal M,j}$ is monotone, hence its minimum and maximum $z_l$-coordinates are achieved at its endpoints. Hence, $\ell_{\mathcal M,j}\leq \sum_{l=1}^{2N+1} |z_l(p^{t_j})-z_l(p^{t_{j-1}})|$ and $\ell_{\mathcal M}\leq \sum_{j=1}^{m}\sum_{l=1}^{2N+1} |z_l(p^{t_j})-z_l(p^{t_{j-1}})|$. Rearranging the right side of the previous inequality, we get $\ell_{\mathcal M}\leq \sum_{l=1}^{2N+1} \sum_{j=1}^{m}|z_l(p^{t_j})-z_l(p^{t_{j-1}})|$. Further, $\sum_{j=1}^{m}|z_1(p^{t_j})-z_1(p^{t_{j-1}})|=\sum_{j=1}^{m}|t_j-t_{j-1}|=t_m-t_0=1$ and, for each $l=2,\dots,2N+1$, $\sum_{j=1}^{m}|z_l(p^{t_j})-z_l(p^{t_{j-1}})| \in O(D\cdot N^2)$. This is because the $z_l$-extent of $\mathcal C_{\mathcal M}$ is at most $D$ (as $D$ is the length of the longest side of $\Delta$ and every vertex stays inside $\Delta$ throughout $\mathcal M$) and because $\mathcal C_{\mathcal M}$ has a number of local minima and maxima with respect to the coordinate $z_l$ which is in $O(N^2)$. It follows that $\ell_{\mathcal M}\leq \sum_{l=1}^{2N+1} \sum_{j=1}^{m}|z_l(p^{t_j})-z_l(p^{t_{j-1}})|\in \sum_{l=1}^{2N+1} O(D\cdot N^2) \in O(D\cdot N^3)$.\end{proof} 

Next, we argue about the possible perturbations of planar straight-line drawings with bounded resolution.  

\begin{claimX}  \label{cl:perturbation}
Let $\Gamma$ be a planar straight-line drawing of a graph $G$ such that the smallest distance between any two separated objects is some value $\delta>0$. Let $\Gamma'$ be any planar straight-line drawing of $G$ obtained from $\Gamma$ by changing each vertex coordinate by at most $\delta/3$. Then $\Gamma'$ is planar. 
\end{claimX}

\begin{proof}
We first prove that any two non-adjacent edges $e_1$ and $e_2$ of $G$ do not cross in $\Gamma'$. For $i=1,2$, the edge $e_i$ is represented by a straight-line segment $s_i$ in $\Gamma$ and by a straight-line segment $s'_i$ in $\Gamma'$. Let~$K$ denote a disk with diameter $0.45 \cdot \delta$. Since $s_1$ and $s_2$ are separated geometric objects in $\Gamma$, the distance between them is at least $\delta$. Hence, the Minkowski sums $R_1:=s_1+K$ and $R_2:=s_2+K$ define non-intersecting regions of the plane, as every point of $R_1$ is at distance at least $\delta-2\cdot (0.45 \cdot \delta)>0$ from every point of $R_2$. Further, since each vertex coordinate in $\Gamma'$ differs by at most $\delta/3$ by the same vertex coordinate in $\Gamma$, it follows that the distance between the positions of a vertex in $\Gamma$ and $\Gamma'$ is at most $\sqrt 2 \cdot \delta/3<0.45 \cdot \delta$. Hence, for $i=1,2$, each end-point of $s'_i$ lies inside $R_i$ and thus the entire segment~$s'_i$ lies inside $R_i$. It follows that $s'_1$ and $s'_2$ do not intersect.  

Analogous proofs show that any two distinct vertices of $G$ do not overlap in $\Gamma'$ and any vertex and any non-incident edge of $G$ do not overlap in $\Gamma'$. 

Finally, we prove that any two adjacent edges $e_1$ and $e_2$ of $G$ do not overlap in $\Gamma'$. If they do, then an end-vertex of one of them, say an end-vertex of $e_1$, overlaps with $e_2$ or overlaps with a vertex of $e_2$, two cases which we already ruled out.  
\end{proof}

We are now ready to prove Theorem~\ref{th:linearize}. We are going to define a sequence $\langle \Psi_0,\Psi_1,\dots,\Psi_k\rangle$ of planar straight-line drawings of $G$ satisfying Properties~(a)--(e). The sequence is initialized by setting $\Psi_0=\Gamma_0$, as required by Property~(a). Suppose now that a sequence $\langle \Psi_0,\Psi_1,\dots,\Psi_j\rangle$ has been defined, for some integer $j\geq 0$, so that Properties~(a)--(d) are satisfied (when restricted to the sequence $\langle \Psi_0,\Psi_1,\dots,\Psi_j\rangle$ constructed so far). We will deal with Property~(e) later.  

While $\Psi_j$ does not coincide with $\Gamma_1$, we add a drawing $\Psi_{j+1}$ to the sequence $\langle \Psi_0,\Psi_1,\dots,\Psi_j\rangle$, so that $\Psi_{j+1}$ belongs to $\mathcal M$ (as required by Property~(b)), so that it is ``closer'' to $\Gamma_1$ in $\mathcal M$ (as required by Property~(c)), and so that the linear morph $\langle \Psi_j,\Psi_{j+1}\rangle$ is planar (as required by Property~(d)). This is done as follows. 

The drawing $\Psi_j$ corresponds to a point $p^{t_j}$ of the FG-curve $\mathcal C_{\mathcal M}$ in $\mathbb R^{2N+1}$. Since the length of the longest side of $\Delta$ is a value $D$, the largest distance between two separated geometric objects in $\Psi_j$ is~$D$, as well. Further, since the resolution of $\Psi_j$ is larger than or equal to $\left(r/n\right)^{c\cdot n}$, the smallest distance between two separated geometric objects in $\Psi_j$ is larger than or equal to $D\cdot \left(r/n\right)^{c\cdot n}$. By Claim~\ref{cl:perturbation}, any straight-line drawing of $G$ obtained from $\Psi_j$ by changing the coordinate of each vertex by at most $D/3\cdot \left(r/n\right)^{c\cdot n}$ is planar. Hence, consider a ball $\mathcal B$ with radius $D/3\cdot \left(r/n\right)^{c\cdot n}$ centered at $p^{t_j}$ in $\mathbb R^{2N+1}$. Let $t_{j+1}\in [0,1]$ be the largest value such that the point $p^{t_{j+1}}$ of $\mathcal C_{\mathcal M}$ lies inside or on the boundary of $\mathcal B$ and let $\Psi_{j+1}$ be the drawing of $G$ in $\mathcal M$ corresponding to $p^{t_{j+1}}$ (thus the sequence $\langle \Psi_0,\Psi_1,\dots,\Psi_j,\Psi_{j+1}\rangle$ satisfies Property~(b)). Since $\mathcal C_{\mathcal M}$ is a continuous curve, we have that $t_{j+1}>t_j$ (thus the sequence $\langle \Psi_0,\Psi_1,\dots,\Psi_j,\Psi_{j+1}\rangle$ satisfies Property~(c)). Further, every point of $\mathcal B$ corresponds to a planar straight-line drawing of $G$ that can be obtained from $\Psi_j$ by changing the coordinate of each vertex by at most $D/3\cdot \left(r/n\right)^{c\cdot n}$. Since the drawings in the linear morph $\langle \Psi_j,\Psi_{j+1}\rangle$ all correspond to points in $\mathcal B$, we have that $\langle \Psi_j,\Psi_{j+1}\rangle$ is planar (thus the sequence $\langle \Psi_0,\Psi_1,\dots,\Psi_j,\Psi_{j+1}\rangle$ satisfies Property~(d)).  

The definition of the sequence $\langle \Psi_0,\Psi_1,\dots,\Psi_k\rangle$ ends when the drawing $\Psi_k$ that is added to the sequence is $\Gamma_1$ (thus the sequence $\langle \Psi_0,\Psi_1,\dots,\Psi_k\rangle$ satisfies Property~(a)). We prove that this happens with $k\in \left(n/r\right)^{O(n)}$. The proof is based on two facts. First, by Claim~\ref{cl:curve-length}, the length of $\mathcal C_{\mathcal M}$ is at most $\gamma\cdot D\cdot N^3$, for some constant $\gamma$. Second, for $j=0,1,\dots,k-2$, the length of the part of $\mathcal C_{\mathcal M}$ between $\Psi_j$ and $\Psi_{j+1}$ is at least $D/3\cdot \left(r/n\right)^{c\cdot n}$. Indeed, $\Psi_j$ and $\Psi_{j+1}$ correspond to points $p^{t_j}$ and $p^{t_{j+1}}$ which are respectively at the center and on the boundary of a ball $\mathcal B$ of radius $D/3\cdot \left(r/n\right)^{c\cdot n}$ (that $p^{t_{j+1}}$ is on the boundary of $\mathcal B$ comes from the continuity of the curve $\mathcal C_{\mathcal M}$ and from the fact that $\Psi_{j+1}$ does not coincide with $\Gamma_1$, since $j\leq k-2$). Hence, the part of $\mathcal C_{\mathcal M}$ connecting  $p^{t_j}$ and $p^{t_{j+1}}$ has length at least $D/3\cdot \left(r/n\right)^{c\cdot n}$ (with the lower bound achieved if this part of $\mathcal C_{\mathcal M}$ is a straight-line segment). Thus, $k-1$ is at most $\gamma\cdot D\cdot N^3$ over $D/3\cdot \left(r/n\right)^{c\cdot n}$, and hence $k\in \left(n/r\right)^{O(n)}$. Thus, the sequence $\langle \Psi_0,\Psi_1,\dots,\Psi_k\rangle$ satisfies Property~(e), which concludes the proof of Theorem~\ref{th:linearize}.

\section{Conclusions and Open Problems} \label{se:conclusions}

In this paper, we have studied the resolution of popular algorithms for the construction of planar straight-line graph drawings and morphs. In fact, with a focus on maximal plane graphs, we discussed the resolution of the drawing algorithm by Floater~\cite{f-psa-97}, which is a broad generalization of Tutte's algorithm~\cite{t-hdg-63}, and of the morphing algorithm by Floater and Gotsman~\cite{fg-hmti-99}.  
Many problems are left open by our research.

\begin{enumerate}
\item The lower bounds on the resolution of F-drawings and FG-morphs presented in Theorems~\ref{th:lower-bound-gt-drawings} and~\ref{th:lower-bound-morph} apply to maximal plane graphs. A major objective is to extend such bounds to $3$-connected plane graphs. Note that the proof of Theorem~\ref{th:lower-bound-gt-drawings}, from which the proof of Theorem~\ref{th:lower-bound-morph} is derived, exploits heavily the fact that the faces of the considered plane graph are delimited by $3$-cycles, hence such an extension requires novel ideas. 		


\item Theorems~\ref{th:lower-bound-gt-drawings} and~\ref{th:upper-bound-gt-drawings} provide an $r\cdot \lambda^{\Theta(n)}$ bound on the resolution of F-drawings. It would be interesting to determine the polynomial function $f(\lambda)$ which is the base of the exponential function $r\cdot (f(\lambda))^n$ representing the worst-case resolution of F-drawings, up to lower-order terms.
Also, when applied to T-drawings, i.e., drawings constructed by Tutte's algorithm~\cite{t-hdg-63}, the lower and upper bounds in Theorems~\ref{th:lower-bound-gt-drawings} and~\ref{th:upper-bound-gt-drawings} do not coincide on the exponent anymore. Namely, Theorem~\ref{th:lower-bound-gt-drawings} gives an $r/n^{O(n)}\subseteq r/2^{O(n \log n)}$ lower bound (as the smallest positive element $\lambda$ of a coefficient matrix is in $1/\Omega(n)$ for graphs with maximum degree in $\Omega(n)$), while Theorem~\ref{th:upper-bound-gt-drawings} gives an $r/2^{\Omega(n)}$ upper bound (as every internal vertex of the graph in the proof of the theorem has degree in $\Theta(1)$, and hence every  element of the coefficient matrix is in $\Theta(1)$). We find it very interesting to close this gap. It is possible (and we believe likely) that this could be achieved by means of an algebraic approach along the lines of the one of~\cite{mrs-sge-11,rg-rsp-96}, as discussed in Section~\ref{se:algebra}. 
\item The lower and upper bounds for the resolution of FG-morphs in Theorems~\ref{th:lower-bound-morph} and~\ref{th:upper-bound-morphing} also leave a gap. First, the dependency on $n$ in the lower bound is $1/n^{O(n)}\subseteq 1/2^{O(n\log n)}$, while the one in the upper bound is $1/2^{\Omega(n)}$. Second, the resolution $r$ of the input drawings appears in the exponential lower bound of Theorem~\ref{th:lower-bound-morph}, while it does not appear in the exponential upper bound of Theorem~\ref{th:upper-bound-morphing}; while it is clear that a dependency on $r$ is needed (as the resolution of the entire morph is clearly upper bounded by the one of the input drawings), it is not clear to us whether $r$ should be part of the exponential function. 
\end{enumerate}

\acknowledgements Thanks to Mirza Klimenta for implementing Floater and Gotsman's algorithm. Further thanks to Daniel Perrucci and Elias Tsigaridas for useful discussions about positive polynomial functions. Finally, thanks to the anonymous reviewers for many interesting comments.

\bibliographystyle{splncs04}
\bibliography{bibliography}

\begin{thebibliography}{10}
\providecommand{\url}[1]{\texttt{#1}}
\providecommand{\urlprefix}{URL }
\providecommand{\doi}[1]{https://doi.org/#1}

\bibitem{aaccdd-hmpgd-17}
Alamdari, S., Angelini, P., Barrera{-}Cruz, F., Chan, T.M., {Da Lozzo}, G., {Di
  Battista}, G., Frati, F., Haxell, P., Lubiw, A., Patrignani, M., Roselli, V.,
  Singla, S., Wilkinson, B.T.: How to morph planar graph drawings. {SIAM}
  Journal on Computing  \textbf{46}(2),  824--852 (2017).
  \doi{10.1137/16M1069171}

\bibitem{a-ramm-02}
Alexa, M.: Recent advances in mesh morphing. Computer Graphics Forum
  \textbf{21}(2),  173--196 (2002). \doi{10.1111/1467-8659.00575}

\bibitem{adflpr-omcd-15}
Angelini, P., {Da Lozzo}, G., Frati, F., Lubiw, A., Patrignani, M., Roselli,
  V.: Optimal morphs of convex drawings. In: Arge, L., Pach, J. (eds.) 31st
  International Symposium on Computational Geometry ({SoCG} '15). LIPIcs,
  vol.~34, pp. 126--140. Schloss Dagstuhl - Leibniz-Zentrum f{\"{u}}r
  Informatik (2015). \doi{10.4230/LIPIcs.SOCG.2015.126}

\bibitem{abc-pd-19}
Arseneva, E., Bose, P., Cano, P., {D'Angelo}, A., Dujmovic, V., Frati, F.,
  Langerman, S., Tappini, A.: Pole dancing: {3D} morphs for tree drawings.
  Journal of Graph Algorithms and Applications  \textbf{23}(3),  579--602
  (2019). \doi{10.7155/jgaa.00503}

\bibitem{br-scdpg-06}
B{\'a}r{\'a}ny, I., Rote, G.: Strictly convex drawings of planar graphs.
  Documenta Mathematica  \textbf{11},  369--391 (2006)

\bibitem{bbd-hmt-19}
Barrera{-}Cruz, F., Borrazzo, M., {Da Lozzo}, G., {Di Battista}, G., Frati, F.,
  Patrignani, M., Roselli, V.: How to morph a tree on a small grid. Discrete \&
  Computational Geometry  \textbf{67}(3),  743--786 (2022).
  \doi{10.1007/S00454-021-00363-8}

\bibitem{bhl-msd-19}
Barrera{-}Cruz, F., Haxell, P.E., Lubiw, A.: Morphing {S}chnyder drawings of
  planar triangulations. Discrete \& Computational Geometry  \textbf{61},
  161--184 (2019). \doi{10.1007/s00454-018-0018-9}

\bibitem{befklow-mpgd3d-23}
Buchin, K., Evans, W.S., Frati, F., Kostitsyna, I., L{\"{o}}ffler, M.,
  Ophelders, T., Wolff, A.: Morphing planar graph drawings through {3D}.
  Computing in Geometry and Topology  \textbf{2}(1),  5:1--5:18 (2023).
  \doi{10.57717/cgt.v2i1.33}

\bibitem{cegl-dgp-12}
Chambers, E.W., Eppstein, D., Goodrich, M.T., L{\"{o}}ffler, M.: Drawing graphs
  in the plane with a prescribed outer face and polynomial area. Journal of
  Graph Algorithms and Applications  \textbf{16}(2),  243--259 (2012).
  \doi{10.7155/jgaa.00257}

\bibitem{cg-cdg-96}
Chrobak, M., Goodrich, M.T., Tamassia, R.: Convex drawings of graphs in two and
  three dimensions (preliminary version). In: Whitesides, S. (ed.) 12th Annual
  Symposium on Computational Geometry ({SoCG} '96). pp. 319--328. {ACM} (1996).
  \doi{10.1145/237218.237401}

\bibitem{ddf-upm-20}
{Da Lozzo}, G., {Di Battista}, G., Frati, F., Patrignani, M., Roselli, V.:
  Upward planar morphs. Algorithmica  \textbf{82}(10),  2985--3017 (2020).
  \doi{10.1007/s00453-020-00714-6}

\bibitem{dpp-htgg-90}
{de Fraysseix}, H., Pach, J., Pollack, R.: How to draw a planar graph on a
  grid. Combinatorica  \textbf{10}(1),  41--51 (1990). \doi{10.1007/BF02122694}

\bibitem{ds-esp-17}
Demaine, E.D., Schulz, A.: Embedding stacked polytopes on a polynomial-size
  grid. Discrete \& Computational Geometry  \textbf{57}(4),  782--809 (2017).
  \doi{10.1007/s00454-017-9887-6}

\bibitem{dett-gd-99}
{Di Battista}, G., Eades, P., Tamassia, R., Tollis, I.G.: Graph Drawing:
  Algorithms for the Visualization of Graphs. Prentice-Hall (1999)

\bibitem{DBLP:conf/gd/BattistaF21}
{Di Battista}, G., Frati, F.: From {T}utte to {F}loater and {G}otsman: {O}n the
  resolution of planar straight-line drawings and morphs. In: Purchase, H.C.,
  Rutter, I. (eds.) 29th International Symposium on Graph Drawing and Network
  Visualization ({GD} '21). LNCS, vol. 12868, pp. 109--122. Springer (2021).
  \doi{10.1007/978-3-030-92931-2\_8}

\bibitem{dlt-pepg-84}
Dolev, D., Leighton, F.T., Trickey, H.: Planar embedding of planar graphs.
  Advances in Computing Research  \textbf{2} (1984)

\bibitem{e-hgd-84}
Eades, P.: A heuristic for graph drawing. Congressus Numerantium
  \textbf{42}(11),  149--160 (1984)

\bibitem{eg-dspg-95}
Eades, P., Garvan, P.: Drawing stressed planar graphs in three dimensions. In:
  Brandenburg, F. (ed.) Symposium on Graph Drawing ({GD} '95). LNCS, vol.~1027,
  pp. 212--223. Springer (1995). \doi{10.1007/BFb0021805}

\bibitem{DBLP:journals/jgaa/EricksonL23}
Erickson, J., Lin, P.: Planar and toroidal morphs made easier. Journal of Graph
  Algorithms and Applications  \textbf{27}(2),  95--118 (2023).
  \doi{10.7155/jgaa.00616}

\bibitem{ekp-ifm-03}
Erten, C., Kobourov, S.G., Pitta, C.: Intersection-free morphing of planar
  graphs. In: Liotta, G. (ed.) 11th International Symposium on Graph Drawing
  ({GD} '03). LNCS, vol.~2912, pp. 320--331. Springer (2003).
  \doi{10.1007/978-3-540-24595-7\_30}

\bibitem{f-psa-97}
Floater, M.S.: Parametrization and smooth approximation of surface
  triangulations. Computer Aided Geometric Design  \textbf{14}(3),  231--250
  (1997). \doi{10.1016/S0167-8396(96)00031-3}

\bibitem{f-pt-98}
Floater, M.S.: Parametric tilings and scattered data approximation.
  International Journal of Shape Modeling  \textbf{4}(3-4),  165--182 (1998).
  \doi{10.1142/S021865439800012X}

\bibitem{f-mvc-03}
Floater, M.S.: Mean value coordinates. Computer Aided Geometric Design
  \textbf{20}(1),  19--27 (2003). \doi{10.1016/S0167-8396(03)00002-5}

\bibitem{fg-hmti-99}
Floater, M.S., Gotsman, C.: How to morph tilings injectively. Journal of
  Computational and Applied Mathematics  \textbf{101}(1),  117--129 (1999).
  \doi{10.1016/S0377-0427(98)00202-7}

\bibitem{fp-nma-07}
Frati, F., Patrignani, M.: A note on minimum-area straight-line drawings of
  planar graphs. In: Hong, S., Nishizeki, T., Quan, W. (eds.) {GD} '07. LNCS,
  vol.~4875, pp. 339--344. Springer (2007). \doi{10.1007/978-3-540-77537-9\_33}

\bibitem{ggt-dof-06}
Gortler, S.J., Gotsman, C., Thurston, D.: Discrete one-forms on meshes and
  applications to {3D} mesh parameterization. Computer Aided Geometric Design
  \textbf{23}(2),  83--112 (2006). \doi{10.1016/j.cagd.2005.05.002}

\bibitem{gjm-27-21}
Groiss, L., J{\"{u}}ttler, B., Mokris, D.: 27 variants of {T}utte's theorem for
  plane near-triangulations and an application to periodic spline surface
  fitting. Computer Aided Geometric Design  \textbf{85},  101975 (2021).
  \doi{10.1016/j.cagd.2021.101975}

\bibitem{i-vfg-14}
Ilinkin, I.: Visualization of {F}loater and {G}otsman's morphing algorithm. In:
  Cheng, S., Devillers, O. (eds.) 30th Annual Symposium on Computational
  Geometry ({SoCG '14}). pp. 94--95. {ACM} (2014).
  \doi{10.1145/2582112.2595649}

\bibitem{DBLP:journals/comgeo/KleistKLSSS19}
Kleist, L., Klemz, B., Lubiw, A., Schlipf, L., Staals, F., Strash, D.:
  Convexity-increasing morphs of planar graphs. Computational Geometry: Theory
  \& Applications  \textbf{84},  69--88 (2019).
  \doi{10.1016/j.comgeo.2019.07.007}

\bibitem{k-fd-13}
Kobourov, S.G.: Force-directed drawing algorithms. In: Tamassia, R. (ed.)
  Handbook on Graph Drawing and Visualization, pp. 383--408. Chapman and
  Hall/CRC (2013)

\bibitem{llw-rb-88}
Linial, N., Lov{\'{a}}sz, L., Wigderson, A.: Rubber bands, convex embeddings
  and graph connectivity. Combinatorica  \textbf{8}(1),  91--102 (1988).
  \doi{10.1007/BF02122557}

\bibitem{mrs-sge-11}
Mor, A.R., Rote, G., Schulz, A.: Small grid embeddings of 3-polytopes. Discrete
  \& Computational Geometry  \textbf{45}(1),  65--87 (2011).
  \doi{10.1007/s00454-010-9301-0}

\bibitem{rg-rsp-96}
Richter-Gebert, J.: Realization Spaces of Polytopes. Springer Berlin,
  Heidelberg (1996)

\bibitem{s-epgg-90}
Schnyder, W.: Embedding planar graphs on the grid. In: Johnson, D.S. (ed.) 1st
  Annual {ACM-SIAM} Symposium on Discrete Algorithms ({SODA} '90). pp.
  138--148. {SIAM} (1990)

\bibitem{DBLP:journals/jgaa/Schulz11}
Schulz, A.: Drawing 3-polytopes with good vertex resolution. Journal of Graph
  Algorithms and Applications  \textbf{15}(1),  33--52 (2011).
  \doi{10.7155/jgaa.00216}

\bibitem{swln-tm-03}
Siu, A.M.K., Wan, A.S.K., Lau, R.W.H., Ngo, C.: Trifocal morphing. In: Banissi,
  E., B{\"{o}}rner, K., Chen, C., Clapworthy, G., Maple, C., Lobben, A., Moore,
  C.J., Roberts, J.C., Ursyn, A., Zhang, J.J. (eds.) 7th International
  Conference on Information Visualization ({IV '03}). pp. 24--29. {IEEE}
  (2003). \doi{10.1109/IV.2003.1217952}

\bibitem{gs-cmcpt-01}
Surazhsky, V., Gotsman, C.: Controllable morphing of compatible planar
  triangulations. {ACM} Transactions on Graphics  \textbf{20}(4),  203--231
  (2001). \doi{10.1145/502783.502784}

\bibitem{t-hdg-63}
Tutte, W.T.: How to draw a graph. Proceedings of the London Mathematical
  Society  \textbf{3}(13),  743--767 (1963). \doi{10.1112/plms/s3-13.1.743}

\bibitem{cpv-tbm-03}
de~Verdi{\`{e}}re, {\'{E}}.C., Pocchiola, M., Vegter, G.: Tutte's barycenter
  method applied to isotopies. Computational Geometry: Theory and Applications
  \textbf{26}(1),  81--97 (2003). \doi{10.1016/S0925-7721(02)00174-8}

\bibitem{warv-asf-02}
Weiss, V., Andor, L., Renner, G., V{\'{a}}rady, T.: Advanced surface fitting
  techniques. Computer Aided Geometric Design  \textbf{19}(1),  19--42 (2002).
  \doi{10.1016/S0167-8396(01)00086-3}

\end{thebibliography}

\end{document}